\documentclass[journal ,12]{IEEEtran}

\usepackage{caption}
\usepackage{pstricks,chngpage,balance}
\usepackage{subcaption}
\usepackage{amsmath,amsfonts,verbatim,amsopn,multirow}

\usepackage{xcolor}
\usepackage{algorithm}
\usepackage{array}
\usepackage{textcomp}
\usepackage{stfloats}
\usepackage{url}
\usepackage{verbatim}
\usepackage{graphicx}
\usepackage{cite}
\usepackage{bm} 
\usepackage{booktabs}
\usepackage{algorithm}
\usepackage{algpseudocode}
\usepackage{algorithm}
\usepackage{amssymb}
\usepackage{algpseudocode}
\usepackage{epsfig}
\usepackage{epstopdf}
\usepackage{amsmath}  
\usepackage[colorlinks=true, linkcolor=blue, citecolor=blue, urlcolor=blue]{hyperref}
\usepackage{pifont} 
\usepackage{booktabs}
\usepackage{amsmath}
\usepackage{caption}
\usepackage{subcaption}
\usepackage{makecell}
\usepackage{amsthm}

\newtheorem{proposition}{Proposition}

\newcommand{\Input}{\item[\textbf{Input:}]}
\newcommand{\Output}{\item[\textbf{Output:}]}

\usepackage[nodisplayskipstretch]{setspace}
\newtheorem{remark}{Remark}

\newcommand{\qa}{{\bf a}}
\newcommand{\qb}{{\bf b}}

\newcommand{\qe}{{\bf e}}
\newcommand{\qf}{{\bf f}}
\newcommand{\qg}{{\bf g}}
\newcommand{\qh}{{\bf h}}

\newcommand{\qn}{{\bf n}}

\newcommand{\qp}{{\bf p}}

\newcommand{\qs}{{\bf s}}

\newcommand{\qu}{{\bf u}}
\newcommand{\qv}{{\bf v}}
\newcommand{\qw}{{\bf w}}
\newcommand{\qx}{{\bf x}}
\newcommand{\qy}{{\bf y}}

\newcommand{\qA}{{\bf A}}

\newcommand{\qG}{{\bf G}}
\newcommand{\qH}{{\bf H}}
\newcommand{\qI}{{\bf I}}
\newcommand{\qJ}{{\bf J}}

\newcommand{\qN}{{\bf N}}

\newcommand{\qR}{{\bf R}}
\newcommand{\qS}{{\bf S}}

\newcommand{\qU}{{\bf U}}
\newcommand{\qV}{{\bf V}}

\newcommand{\qY}{{\bf Y}}

\newcommand{\tens}[1]{\boldsymbol{\mathcal{#1}}}

\newcommand{\pcj}{p^{\mathtt{c}}_{j}}
\newcommand{\pci}{p^{\mathtt{c}}_{i}}

\newcommand{\vrmT}{\qv_r^{(m)T}}
\newcommand{\vtmT}{\qv_t^{(m)T}}

\newcommand{\wpk}{\qw^{\mathtt{p}}_k}
\newcommand{\wcj}{\qw^{\mathtt{c}}_{j}}
\newcommand{\wcjH}{\qw^{\mathtt{c}H}_{j}}

\newcommand{\SINR}{\mathrm{SINR}}

\newcommand{\Thetat}{\mathbf{\Theta}_t}
\newcommand{\Thetar}{\mathbf{\Theta}_r}

\newcommand{\Kc}{\mathcal{K}_{c}}
\newcommand{\Kp}{\mathcal{K}_{p}}

\newcommand{\Ex}{\mathbb{E}}
\newcommand{\wkc}{\qw_{k_c}^{\mathtt{c}}}
\newcommand{\wkp}{\qw_{k_p}^{\mathtt{p}}}

\newcommand{\ppkp}{p^{\mathtt{p}}_{k_p}}
\newcommand{\pckc}{p^{\mathtt{c}}_{k_c}}
\newcommand{\pck}{p^{\mathtt{c}}_{k}}

\newcommand{\cmark}{\ding{51}} 
\newcommand{\xmark}{\ding{55}} 
\DeclareMathOperator{\tr}{\mathrm{tr}}

\begin{document}

\title{Hybrid STAR-RIS Architecture for Joint Localization, Communication, and Power Transfer}

\author{Haoran Ni,~\IEEEmembership{Graduate Student Member,~IEEE}, Mohammadali Mohammadi,~\IEEEmembership{Senior Member,~IEEE,}\\ Xidong Mu,~\IEEEmembership{Member,~IEEE,} Hien Quoc Ngo,~\IEEEmembership{Fellow,~IEEE,} and Michail Matthaiou,~\IEEEmembership{Fellow,~IEEE}\\
\thanks{ This work was supported by the UK Engineering and Physical Sciences
Research Council (EPSRC) grant EP/X04047X/2 for TITAN Telecoms Hub. This work of Xidong Mu was supported by the Royal Society under grant RG\textbackslash R1\textbackslash 251180.
The work of H. Q. Ngo was supported by the a research grant
from the Department for the Economy Northern Ireland under the US-Ireland
R\&D Partnership Programme, and MULTIPLY-6G project that has received funding from the Smart Networks and Services Joint Undertaking (SNS JU) under the European Union’s Horizon Europe research and innovation programme under Grant Agreement No 101293106. This work was also supported by the European
Research Council (ERC) under the European Union’s Horizon 2020 research
and innovation programme (grant agreement No. 101001331).}

\thanks{The authors are with the Centre for Wireless Innovation (CWI), Queen's University Belfast, BT3 9DT Belfast, U.K.
(email:\{hni02, m.mohammadi, x.mu, hien.ngo, m.matthaiou\}@qub.ac.uk).  }
\thanks{ Parts of this paper appeared at the 2025 IEEE SPAWC conference~\cite{Ni:Hybrid-STAR-RIS:SPAWC:2025}.
}}\normalsize

\maketitle

\begin{abstract}
We propose a hybrid simultaneously transmitting and reflecting reconfigurable intelligent surface (STAR-RIS) architecture with dynamically switched active and passive elements to support joint  localization, communication, and wireless power transfer (WPT). We first pursue a parallel factor analysis with the alternating least squares (PARAFAC-ALS)-based tensor decomposition approach that decouples the base station (BS)–reconfigurable intelligent surface (RIS) and RIS–user channels, thereby enabling low-overhead channel acquisition. Based on this, we formulate a system energy efficiency (EE) maximization problem, subject to the spectral efficiency (SE) requirements of communication users, sensing signal-to-interference-plus-noise ratio  constraints, and the nonlinear energy harvesting requirements of energy-harvesting users. The  optimization problem is nonconvex since the transmit power allocation, STAR-RIS coefficients, and active/passive mode assignments are tightly coupled in both the objective and constraints. We address this issue by alternating between two subproblems, and solving them via fractional programming, successive convex approximation and a multi-seed greedy strategy employed as an initialization step. Numerical results demonstrate that selectively activating a small, well-chosen subset of STAR-RIS elements achieves $1.5\times$ to $3\times$ EE improvements compared with fully passive/active architectures, while satisfying communication, sensing, and power-transfer requirements.
\end{abstract}


\vspace{-1em}
\section{Introduction}
\IEEEPARstart{T}{he}  emergence of simultaneously  transmitting and reflecting  reconfigurable intelligent surfaces (STAR-RISs) marks a significant milestone in the evolution of sixth-generation (6G) wireless networks. Unlike conventional reconfigurable intelligent surfaces (RISs), which operate solely in reflection mode and thus cannot assist users or sensing targets located behind the surface, STAR-RISs enable full-space coverage by supporting controllable energy splitting between transmission and reflection regions. This capability introduces additional spatial degrees of freedom, paving the way for advanced designs in communication, localization, integrated sensing and communication (ISAC), and wireless power transfer (WPT) \cite{Liu:STAR:WCOM:2021,Ahmed:STAR-RIS:IOT:2023,Sui:STAR-RIS:TWC:2025}.

The evolution toward 6G wireless networks is driving a paradigm shift from conventional communication-centric architectures to multifunctional platforms \cite{Matthaiou:6G-Phy:COMMAG:2021} capable of simultaneously supporting communication, sensing, localization \cite{Lyu:CRB-RIS:IOT:2024}, and WPT services \cite{Mohammadi_cellfree_swipt:TWC:2025}. The practical relevance of such multifunctional operation is already evident in several emerging 6G application scenarios. For instance, mobility-aware Internet of Things (IoT) networks require the seamless integration of communication, sensing, and energy transfer functionalities to support energy-constrained devices operating in highly dynamic environments \cite{fan2025mobility}. Similarly, low-altitude economy networks involving uncrafted aerial vehicles (UAVs) and other aerial platforms rely on the joint provision of communication, sensing, and WPT services to ensure reliable connectivity, environmental awareness, and sustainable operation \cite{zhang2026intelligent}. These representative use cases underscore that the convergence of communication, sensing, and WPT is not merely a theoretical research direction but a fundamental requirement for meeting the diverse service demands envisioned for future 6G networks. In response to this trend, integrated sensing, communication, and powering (ISCAP) and integrated sensing, communication, and power transfer (ISCPT) have emerged as promising frameworks for the efficient integration of these functionalities through the shared utilization of spectrum, energy resources, and hardware infrastructures \cite{chen2025integrated,li2024integrating}.

\vspace{-0.3em}
\subsection{ Related Literature}
Despite these advances, the performance of integrated sensing, communication, and WPT systems remains highly dependent on the wireless propagation environment. In particular, sensing relies on reliable target illumination and echo reception, whereas WPT requires favorable channel conditions for efficient energy delivery to wireless devices. Consequently, non-line-of-sight (NLoS) propagation, severe path loss, and blockage can significantly impair the performance of sensing, communication, and WPT functionalities. To address these challenges, RISs and STAR-RISs have emerged as promising technologies for controlling the wireless propagation environment. By intelligently reconfiguring electromagnetic wave propagation, these technologies can enhance signal coverage, improve sensing performance, and facilitate efficient wireless energy delivery to communication, sensing, and energy-harvesting terminals~\cite{Ren:UAV-WPT:IOT:2023,Zhao:RIS-WPT:SYS:2021,MohammadiRIS_SWIPT_opt:TCOM:2025,Liu:ISAC:JSAC:2022,Zhang:WPT:TIE:2019,Hua:BD-RIS:TCOM:2025}.

\subsubsection{RIS-assisted  WPT, ISAC, and integrated ISAC-WPT systems}
Fu \textit{et al.}\cite{10155421} investigated a multi-active/passive-RIS-enabled simultaneous wireless information and power transfer (SWIPT) system and analyzed the impact of active RIS deployment on WPT performance under different network configurations.  Kang \textit{et al.} \cite{kang2024active} proposed a hybrid active-passive RIS architecture that combines the signal amplification capability of active elements with the energy-efficient reflection of passive elements, thereby enhancing wireless coverage and system capacity while accounting for practical power-consumption constraints.  Ata \textit{et al.} \cite{ata2025reflect} applied the hybrid active-passive RIS architecture to SWIPT systems, where passive elements are used for information transmission, while active elements amplify the incident signal to improve energy harvesting performance at low-power receivers. Yao \textit{et al.} \cite{yao2025hybrid} further extended the hybrid RIS architecture to secure ISAC systems. They jointly optimized multiple hybrid RISs, the transmit waveform, receive filters, and user signal-to-noise ratio (SNR) to improve extended target sensing performance and satisfy secure communication requirements. Lin \textit{et al.} \cite{lin2025joint} investigated a hybrid RIS-assisted ISAC system and jointly optimized the active/passive mode selection and beamforming design to maximize the minimum sensing beampattern gain among multiple targets. Finally, RIS-assisted ISAC-WPT systems have been proposed to intelligently reconfigure the wireless environment and enhance blocked links \cite{zhang2023performance}.

\subsubsection{STAR-RIS-assisted ISAC, WPT, and integrated ISAC-WPT systems}
Despite their significant potential, conventional RIS architectures are inherently constrained to operate in either the reflection or transmission half-space, which limits their achievable coverage and deployment flexibility. To overcome this limitation, STAR-RIS has emerged as a promising technology capable of simultaneously transmitting and reflecting incident signals, thereby enabling full-space wireless coverage and providing greater flexibility for supporting multifunctional 6G services~\cite{Mu:STAR:TWC:2022}. Motivated by these advantages, recent studies have increasingly explored STAR-RIS-enabled communication, sensing, and WPT frameworks. The STAR-RIS joint beamforming and coefficient optimization can significantly improve detection accuracy under communication constraints~\cite{Wang:TWC:2023}. In cluttered environments, the integration of signature-sequence modulation with passive STAR-RIS beamforming has been shown to enable reliable multi-target sensing~\cite{Zhang:STAR-RIS-ISAC:TCOM:2025}. Moreover, passive STAR-RIS-assisted integrated sensing, computing, and communication has been investigated for robotic networks, leveraging alternating optimization (AO) to jointly update the base station (BS) beamformer and STAR-RIS coefficients~\cite{Li:STAR-RIS-ISCC:IOT:2024}. Collectively, these works highlight that STAR-RISs extend conventional RIS-based ISAC by allowing independent control of transmitted and reflected signals, thereby offering greater flexibility for sensing-oriented system design.   In~\cite{Li:STAR-RIS-SWIPT:TWC:2024}, a passive STAR-RIS–enabled SWIPT framework was studied, where a multi-antenna access point served both communication users (CUs) and energy-harvesting users (EUs). Zhu \textit{et al.} \cite{Zhu:STAR-RIS-SWIPT:TWC:2025} compared passive and active STAR-RIS architectures and demonstrated that jointly optimizing the access point beamforming and STAR-RIS configuration can reduce the power consumption, with active surfaces being advantageous for small apertures, while passive designs become more efficient under larger apertures or stricter power budgets. The authors of \cite{Zhu:Robust-RA-STAR:TWC:2024} showed that STAR-RISs can enhance the rate–energy trade-off in SWIPT systems under imperfect channel state information (CSI) knowledge, outperforming conventional RISs across energy-splitting and time-switching strategies. Finally, \cite{Yaswanth:TCCN:2025} proposed an active STAR-RIS-assisted ISAC SWIPT framework that jointly optimizes the beamforming and RIS configurations under energy harvesting (EH) constraints.

\begin{table*}[t]
\centering
\scriptsize
\renewcommand{\arraystretch}{1.2}
\setlength{\tabcolsep}{3.5pt}
\caption{Comparison of our work with representative prior studies}
\vspace{-0.7em}
\label{tab:comparison}
\resizebox{\textwidth}{!}{
\begin{tabular}{l*{21}{c}}
\toprule
\textbf{Contributions} 
& \textbf{This} 
& {\cite{Sui:STAR-RIS:TWC:2025}}
& {\cite{Lyu:CRB-RIS:IOT:2024}}
& {\cite{Mohammadi_cellfree_swipt:TWC:2025}}
& {\cite{Ren:UAV-WPT:IOT:2023}}
& {\cite{Zhao:RIS-WPT:SYS:2021}}
& {\cite{Hua:BD-RIS:TCOM:2025}}
& {\cite{kang2024active}}
& {\cite{ata2025reflect}}
& {\cite{yao2025hybrid}}
& {\cite{lin2025joint}}
& {\cite{zhang2023performance}}
& {\cite{Mu:STAR:TWC:2022}}
& {\cite{Wang:TWC:2023}}
& {\cite{Zhang:STAR-RIS-ISAC:TCOM:2025}}
& {\cite{Li:STAR-RIS-ISCC:IOT:2024}}
& {\cite{Li:STAR-RIS-SWIPT:TWC:2024}}
& {\cite{Zhu:STAR-RIS-SWIPT:TWC:2025}}
& {\cite{Zhu:Robust-RA-STAR:TWC:2024}}
& {\cite{Yaswanth:TCCN:2025}}
& {\cite{Song:Active-STAR-RIS:TAP:2025}}
\\
& \textbf{paper} 
& & & & & & & & & & & & & & & & & & & &
\\
\midrule

STAR-RIS                     
& \cmark 
& \cmark 
& \xmark 
& \xmark 
& \xmark 
& \xmark 
& \xmark 
& \xmark 
& \xmark 
& \xmark 
& \xmark 
& \xmark 
& \cmark 
& \cmark 
& \cmark 
& \cmark 
& \cmark 
& \cmark 
& \cmark 
& \cmark 
& \cmark 
\\

Hybrid (active and passive elements) 
& \cmark 
& \xmark 
& \xmark 
& \xmark 
& \xmark 
& \xmark 
& \xmark 
& \cmark 
& \cmark 
& \cmark 
& \cmark 
& \xmark 
& \xmark 
& \xmark 
& \xmark 
& \xmark 
& \xmark 
& \xmark 
& \xmark 
& \xmark 
& \cmark 
\\

ISAC                         
& \cmark 
& \xmark 
& \cmark 
& \xmark 
& \xmark 
& \xmark 
& \xmark 
& \xmark 
& \xmark 
& \cmark 
& \cmark 
& \cmark 
& \xmark 
& \cmark 
& \cmark 
& \cmark 
& \xmark 
& \xmark 
& \xmark 
& \cmark 
& \cmark 
\\

SWIPT                        
& \cmark 
& \xmark 
& \xmark 
& \cmark 
& \xmark 
& \xmark 
& \cmark 
& \xmark 
& \cmark 
& \xmark 
& \xmark 
& \cmark 
& \cmark 
& \xmark 
& \xmark 
& \xmark 
& \cmark 
& \xmark 
& \xmark 
& \cmark 
& \xmark 
\\

Non-linear energy harvesting model          
& \cmark 
& \xmark 
& \xmark 
& \cmark 
& \xmark 
& \xmark 
& \cmark 
& \xmark 
& \xmark 
& \xmark 
& \xmark 
& \xmark 
& \xmark 
& \xmark 
& \xmark 
& \xmark 
& \xmark 
& \xmark 
& \xmark 
& \xmark 
& \xmark 
\\

Power allocation             
& \cmark 
& \xmark 
& \cmark 
& \cmark 
& \cmark 
& \cmark 
& \cmark 
& \cmark 
& \cmark 
& \cmark 
& \xmark 
& \cmark 
& \cmark 
& \xmark 
& \cmark 
& \cmark 
& \cmark 
& \cmark 
& \cmark 
& \xmark 
& \cmark 
\\

Active/passive element mode selection 
& \cmark 
& \xmark 
& \xmark 
& \xmark 
& \xmark 
& \xmark 
& \xmark 
& \xmark 
& \xmark 
& \xmark 
& \cmark 
& \xmark 
& \xmark 
& \xmark 
& \xmark 
& \xmark 
& \xmark 
& \xmark 
& \xmark 
& \xmark 
& \xmark 
\\

\bottomrule
\end{tabular}
}
\vspace{-1.5em}
\end{table*}

\vspace{-0.3em}
\subsection{Motivations and Contributions}
Although  fully passive/active STAR-RISs have been investigated in existing works, they both suffer from inherent limitations. Specifically, passive STAR-RISs, dominant in existing literature, suffer from severe cascaded fading, while  active STAR-RISs incur high power consumption and hardware noise~\cite{Song:Active-STAR-RIS:TAP:2025}. To this end, hybrid STAR-RISs offer a promising compromise and can underpin ISAC and WPT \cite{Vu:STAR-RIS-NOMA:COML:2025,Yigit:Hybrid-STAR-RIS-ISAC:TCOM:2025}. However, the potential of hybrid STAR-RISs to jointly support localization, communication, and WPT has not yet been fully investigated. This motivates the joint optimization of hybrid STAR-RIS configurations and resource allocation to balance the signal enhancement provided by active elements with the energy efficiency of passive elements. Developing such an integrated multi-functional framework remains challenging due to the complex interactions among these services.  First, developing a unified system model is challenging, as multiple functionalities must be jointly supported within a single architecture with well-defined performance metrics, such as quality of service (QoS) and EE. Second, channel estimation becomes challenging due to the cascaded channels formed by active and passive STAR-RIS elements.
As noted in~\cite{Liu:STAR:WCOM:2021}, these cascaded characteristics make accurate full CSI acquisition with low overhead particularly difficult. Consequently, some recent works adopted sensing models without dedicated target precoding or assumed idealized target-channel knowledge, thereby limiting practicality~\cite{Wang:TWC:2023}. Third, while active elements enhance link strength, they introduce additional noise and higher power consumption, necessitating a balance between performance and EE in hybrid STAR-RIS systems. Finally, jointly optimizing transmit power and STAR-RIS coefficients results in highly non-convex optimization problems that typically require advanced approximations or decomposition methods.

To this end, we herein design a novel system architecture, where the STAR-RIS comprises a mixture of active and passive elements, jointly serving sensing targets, CUs, and EUs. We depart from conventional fully passive or fully active STAR-RIS designs by investigating a hybrid STAR-RIS architecture with dynamically scheduled active and passive elements. Unlike conventional STAR-RIS designs that rely exclusively on either fully passive or fully active elements, we investigate a hybrid STAR-RIS architecture with dynamically scheduled active and passive elements. In addition, low-cost sensing units are mounted on the STAR-RIS to collect echo signals for  processing. The main contributions of this paper are summarized below and are also explicitly compared to the representative prior works in Table \ref{tab:comparison}.
\vspace{-0.2em}
\begin{itemize}
\item We propose a channel estimation approach based on a three-way tensor representation of uplink (UL) pilots. We design a parallel factor analysis (PARAFAC) alternating least squares (ALS) algorithm that efficiently decouples the cascaded BS-RIS-user channels without additional overhead. The factor matrices provide accurate estimates of the individual BS–STAR-RIS and STAR-RIS–user channels.

\item We develop a target localization algorithm, derive the sensing signal-to-interference-plus-noise ratio (SINR), and present a closed-form Cramér--Rao bound (CRB). Furthermore, we introduce a joint estimator that combines two-dimensional (2D) multiple signal classification (MUSIC) and frequency modulated continuous wave (FMCW) beat-frequency analysis to reliably recover both angular and distance information.

\item We formulate an EE maximization problem that jointly optimizes the transmit power, STAR-RIS coefficients, and mode scheduling for hybrid elements, subject to sensing threshold constraints, communication QoS requirements, and nonlinear EH constraints. To address the resulting complex non-convex optimization problem, we decompose it into two  subproblems. The non-convexity arises from the tight coupling among transmit power allocation, STAR-RIS coefficients, and active/passive mode assignments in the objective and constraints. We solve the two subproblems in an AO manner, using fractional programming, successive convex approximation (SCA) and a multi-seed greedy strategy for the active/passive mode assignment to ensure convergence and scalability.

\item Numerical results confirm the accuracy of the proposed channel estimation and localization algorithms. Furthermore, our analysis demonstrates the performance gains achieved by our optimization framework compared to benchmark models with random active-element selection and fully passive STAR-RIS configurations. Specifically, the proposed scheme achieves up to $6.8$-fold improvement over the fully passive design and $3.3$-fold over the hybrid design with random element selection.
\end{itemize}
\vspace{-1em}
\subsection{Organization and Notation}
The rest of the paper is organized as follows: Section \ref{Sec:system model} presents the hybrid STAR-RIS-enabled joint localization, communication, and WPT system. Section \ref{sec:Transmission Protocol} introduces the transmission protocol, while Section \ref{sec:Performance} defines performance metrics. Section \ref{sec: EE optimization} formulates the EE maximization problem and the AO framework. Section \ref{sec:Numerical} presents simulation results, and Section \ref{sec:Conclusion} concludes the paper.

\textit{Notation:} We use bold lowercase (uppercase) letters to denote vectors (matrices); the superscripts $(\cdot)^{T}$, $(\cdot)^{H}$, and $(\cdot)^{*}$ stand for transpose, Hermitian, and complex conjugate, respectively; the notation $\|\cdot\|_F^2$ denotes the squared Frobenius norm;
$[\qA]{:,n}$ and $[\qA]_{n,:}$ denote the $n$-th column and $n$-th row of matrix $\qA $;
$\mathbb{B}^{N \times N}$ denotes the set of $N \times N$ binary matrices;
$\mathbf{I}_N$ denotes the $N\times N$ identity matrix; $\mathrm{diag}\{\mathbf{a}\}$ denotes a diagonal matrix with diagonal elements from $\mathbf{a}$;
$\tr(\cdot)$ denotes the trace operator;  $\mathrm{vec}(\cdot)$ denotes the vectorization operator stacking matrix columns into a vector; the outer, Khatri–Rao, and Kronecker products are denoted by $\circ$, $\odot$, and $\otimes$, respectively; $\Re\{\cdot\}$ and $\Im\{\cdot\}$ denote real and imaginary parts, respectively. Moreover, $\mathcal{CN}(0,\sigma^2)$ denotes a complex circularly symmetric Gaussian random variable with variance $\sigma^2$. Finally,  
$\mathbb{E}\{\cdot\}$ stands for the statistical expectation. 

\section{system model}\label{Sec:system model}
We consider a hybrid STAR-RIS-enabled joint localization, communication, and WPT system. The BS, equipped with a uniform linear array (ULA) of $M$ antennas, serves $K_c$ CUs, $K_p$ EUs, and one single target. Let $\mathcal{K}_c = \{1,\ldots,K_c\}$ and $\mathcal{K}_p = \{1,\ldots,K_p\}$ denote the index sets of CUs and EUs, respectively. The STAR-RIS is implemented as a uniform planar array (UPA) with $N$ elements at half-wavelength spacing. To mitigate large-scale fading in echo processing, a low-cost UPA sensor with $N_s$ elements is integrated into the STAR-RIS~\cite{Wang:TWC:2023}. The target for localization is assumed to be located in the reflective region.
\begin{figure}[t]
\vspace{-1em}
    \centering
\includegraphics[width=0.45\textwidth]{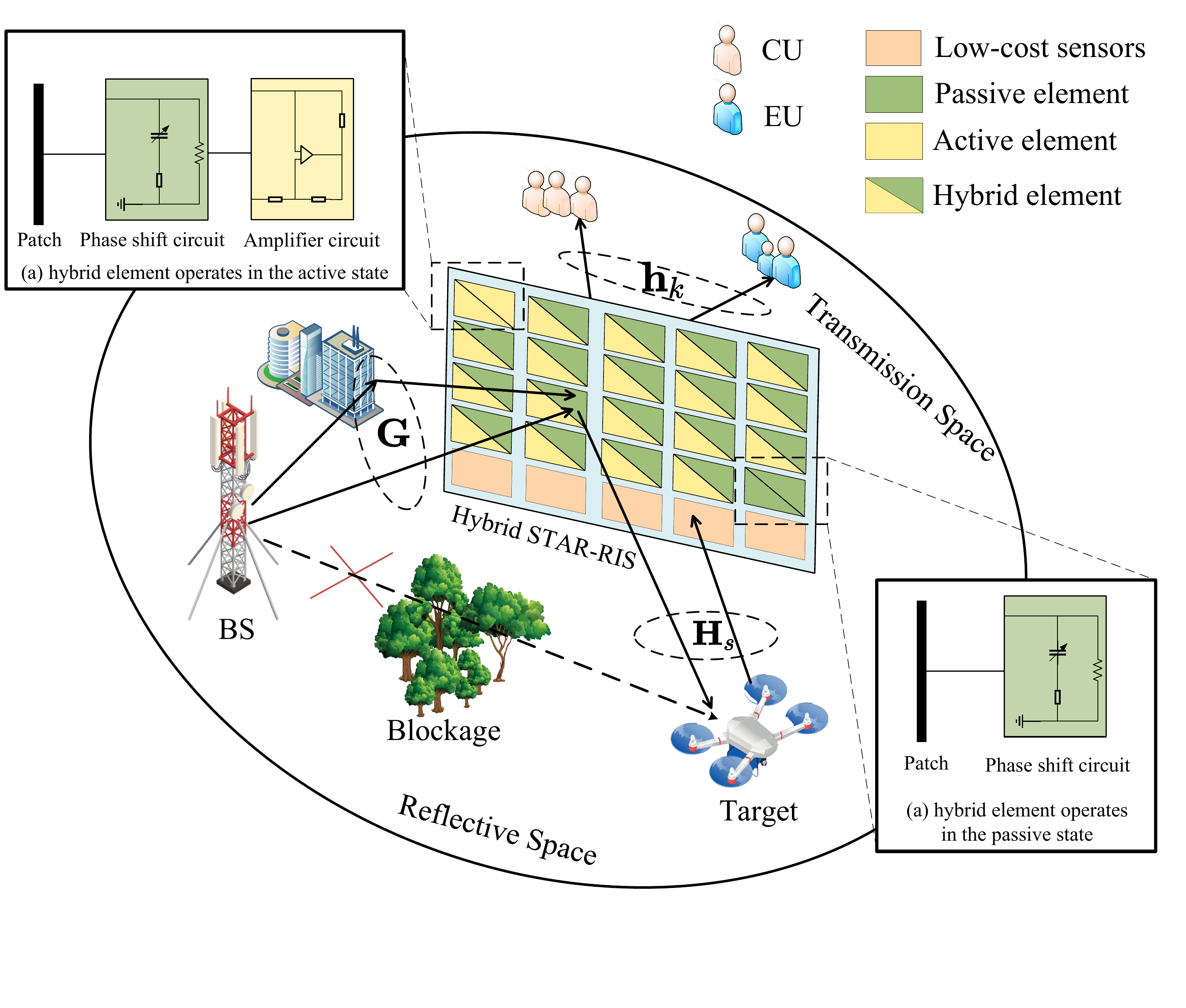}
\vspace{-2.3em}
    \caption{\small Hybrid STAR-RIS architecture for joint localization, communication, and WPT.}
    \label{fig:enter-label}
    \vspace{0em}
\end{figure}

\vspace{-1em}
\subsection{Hybrid Active-Passive STAR-RIS Model}
We set up an additional amplification layer between the reflection and transmission layers of the conventional STAR-RIS to enhance both the reflected and transmitted signals, where its operation is controlled in an ON/OFF manner. Specifically, each hybrid STAR-RIS element is equipped with a phase-shift circuit and a controllable switching circuit connected to an amplification branch \cite{Yigit:Hybrid-STAR-RIS-ISAC:TCOM:2025}. When the switch is turned on, the incident signal is routed through the amplifier, enabling active operation with both amplitude and phase control. When the switch is turned off, the amplification branch is bypassed, and the element operates as a passive STAR-RIS element with phase-only control \cite{kang2024active}. We denote $\mathbf{\Theta}_t \in\mathbb{C}^{N \times N}$ and $\mathbf{\Theta}_r \in\mathbb{C}^{N \times N}$ as the matrices of the transmission coefficients and reflection coefficients, respectively, which are modeled as,
\begin{equation}
    \mathbf{\Theta}_{k} =\operatorname{diag}\{\bm{\phi}_k\}= \operatorname{diag} \left\{ \beta_{1,k} e^{j \theta_{1,k}}, \ldots, \beta_{N,k} e^{j \theta_{N,k}} \right\},
\end{equation}
where $\theta_{n,k}\in(0,2\pi]$ and $\beta_{n,k}\in(0,1]$ are respectively the phase shift and amplitude of the $n$-th element of the STAR-RIS, while \( k \in \{t, r\} \) represents the transmission or reflection side. The ratio of active STAR-RIS elements to the total number of elements is defined as $\varrho = \frac{N_{\text{act}}}{N}$, with $0 \leq \varrho \leq 1$. We denote $\bm\alpha \in\mathbb{C}^{N \times 1} $ as the binary mode scheduling vector for the STAR-RIS elements, where $\alpha_n = 1$ indicates active mode and $\alpha_n = 0$ indicates passive mode. The positions of active elements are predefined in $\mathcal{A} \subset \{1,2,\ldots,N\}$ with $|\mathcal{A}| = N_{\mathrm{act}}$. Accordingly, the hybrid STAR-RIS mode scheduling across all elements is represented by the diagonal matrix as  
\begin{equation}
    \mathbf{A} = \operatorname{diag} \{\alpha_1, \alpha_2, \dots, \alpha_N\} \in \mathbb{B}^{N \times N}.
\end{equation} 
Moreover, we define
$\boldsymbol{\Upsilon}_{k} \triangleq \mathbf{A} \circ \mathbf{\Theta}_{k}$ and $ \mathbf{\Psi}_{k} \triangleq (\mathbf{I}_N - \mathbf{A}) \circ \mathbf{\Theta}_{k}$ to denote the diagonal active and passive coefficient matrices, respectively. According to the law of energy, each passive element  must satisfy \(|\beta_{n,t}|^2 + |\beta_{n,r}|^2 \le 1\)\cite{xu2021star}. For each active element \(n \in \mathcal{A}\), the magnitude of the complex coefficient satisfies \(|\beta_{n,k}| = \rho \leq \rho_{\max}\), where \(\rho_{\max}\) denotes the maximum amplification gain provided by the active load.

\vspace{-0.7em}
\subsection{CU and EU Channel Models}
We assume a Ricean fading channel model for all communication links. Specifically, the channel from the BS to the STAR-RIS and the channel from the STAR-RIS to the $k$-th user can be modeled as 
\begin{align}
    \qG &= \sqrt{\beta_{BR}/(1+\kappa)}\big( 
    \sqrt{\kappa} \qG^{\mathrm{LoS}} + 
     \qG^{\mathrm{NLoS}} 
    \big), 
\end{align}
and
\vspace{-0.5em}
\begin{align}    
    \qh_k &= \sqrt{\beta_{RU,k}/(1+\kappa)}\Big( 
    \sqrt{\kappa} \qh_k^{\mathrm{LoS}} + 
     \qh_k^{\mathrm{NLoS}} 
    \Big),
\end{align}
respectively, where \( {\beta_{BR}} \) and \( {\beta_{RU,k}} \) denote the large-scale fading coefficients between the BS and the STAR-RIS, and between the STAR-RIS and user $k\in\{\Kc,\Kp\}$, respectively. 
The parameter \( \kappa \) is the Ricean factor representing the power ratio between the line-of-sight (LoS) and non-line-of-sight (NLoS) components. 
The terms \( \qG^{\mathrm{NLoS}} \in \mathbb{C}^{N \times M} \) and \( \qh_k^{\mathrm{NLoS}} \in \mathbb{C}^{N\times 1} \) represent the NLoS components, whose entries are independent \( \mathcal{CN}(0,1) \) random variables (RVs). The LoS components of the channels are 
\begin{align}
    \qG^{\text{LoS}} &= \mathbf{b}_{\text{BR}}(\theta_{\text{A}}^{\text{BR}}, \psi_{\text{A}}^{\text{BR}})\, 
    \mathbf{a}_{\text{BR}}^{T}(\theta_{\text{D}}^{\text{BR}}), \\
    \qh_k^{\text{LoS}} &= \mathbf{a}_{\text{RU}}(\theta_{k,\text{D}}^{\text{RU}}, \psi_{k,\text{D}}^{\text{RU}}),
\end{align}
where \( \theta_{\text{D}}^{\text{BR}} \) denotes the angle of departure (AoD) from the BS; 
\( \theta_{\text{A}}^{\text{BR}} \) and \( \psi_{\text{A}}^{\text{BR}} \) denote the azimuth and elevation angles of arrival (AoA) at the STAR-RIS from the BS; \( \theta_{k,\text{D}}^{\text{RU}} \), \( \psi_{k,\text{D}}^{\text{RU}} \) denote the azimuth and elevation angles of departure from the STAR-RIS to user \( k \). Let \(N_h\) and \(N_v\) represent the numbers of STAR-RIS elements placed along the horizontal and vertical axes, respectively, forming a \(N=N_h \times N_v\) UPA.  We denote the corresponding element indices by \(n_h=1,\ldots,N_h\) and \(n_v=1,\ldots,N_v\).
Following \cite{Gao:Hybrid-Precoding:JSAC:2016}, we define $\mathbf{a}_{\text{RU}}(\theta_{k,\text{D}}^{\text{RU}}, \psi_{k,\text{D}}^{\text{RU}})\triangleq\frac{1}{\sqrt{N}}\mathbf{a}_{\text{RU},h}(\theta_{k,\text{D}}^{\text{RU}}, \psi_{k,\text{D}}^{\text{RU}})\otimes\mathbf{a}_{\text{RU},v}(\theta_{k,\text{D}}^{\text{RU}}, \psi_{k,\text{D}}^{\text{RU}})$ with $[\mathbf{a}_{\text{RU},h}(\theta_{k,\text{D}}^{\text{RU}}, \psi_{k,\text{D}}^{\text{RU}})]_{n_h}\triangleq e^{-j\pi(n_h-1)\sin\theta_{k,\text{D}}^{\text{RU}}\sin\psi_{k,\text{D}}^{\text{RU}} }$ and $[\mathbf{a}_{\text{RU},v}(\theta_{k,\text{D}}^{\text{RU}}, \psi_{k,\text{D}}^{\text{RU}})]_{n_v}\triangleq e^{-j\pi(n_v-1)\cos\psi_{k,\text{D}}^{\text{RU}}}$. Similarly,  
$\mathbf{b}_{\text{BR}}(\theta_{\text{A}}^{\text{BR}}, \psi_{\text{A}}^{\text{BR}})
\triangleq \frac{1}{\sqrt{N}}\mathbf{b}_{\text{BR},h}(\theta_{\text{A}}^{\text{BR}}, \psi_{\text{A}}^{\text{BR}})
\otimes \mathbf{b}_{\text{BR},v}(\theta_{\text{A}}^{\text{BR}}, \psi_{\text{A}}^{\text{BR}}),$
with $[\mathbf{b}_{\text{BR},h}(\theta_{\text{A}}^{\text{BR}}, \psi_{\text{A}}^{\text{BR}})]_{n_h}
\triangleq e^{-j\pi (n_h-1)\sin\theta_{\text{A}}^{\text{BR}}\sin\psi_{\text{A}}^{\text{BR}}},$
and $[\mathbf{b}_{\text{BR},v}(\theta_{\text{A}}^{\text{BR}}, \psi_{\text{A}}^{\text{BR}})]_{n_v}
\triangleq e^{-j\pi (n_v-1)\cos\psi_{\text{A}}^{\text{BR}}}.$  Finally, the BS ULA steering vector is given by  $\mathbf{a}_{\text{BR}}(\theta_{\text{D}}^{\text{BR}})\triangleq \frac{1}{\sqrt{M}}
    [ 1,\ e^{-j \pi \cos\theta_{\text{D}}^{\text{BR}}},\ \ldots,\ e^{-j \pi (M-1) \cos\theta_{\text{D}}^{\text{BR}}} ]^{T}.$ The BS, RIS, and sensor arrays are all configured with a half-wavelength inter-element spacing, i.e., $d=\lambda/2$, where $\lambda$ denotes the carrier wavelength.

\vspace{-1em}
\subsection{Sensing Channel Model}
The signal transmitted by the BS is first reflected by the STAR-RIS toward the target, and the echo is then captured by the low-cost UPA sensor integrated on the STAR-RIS. Thus, the sensing channel can be modeled as
\begin{equation}
\qH_s=\beta_s\qb_s(\theta,\psi)\ \qa^T_\text{RT}(\theta_{\text{D}}^{\text{RT}}, \psi_{\text{D}}^{\text{RT}})\mathbf{\Theta}_r\qG,
\end{equation}
where  
$\qH_{s}\in\mathbb{C}^{N_{s}\times M}$ is the composite sensing-echo channel from the BS to the STAR-RIS sensor array, while \(\theta\) and \(\psi\) denote the elevation and azimuth angles of the target with respect to the 
STAR-RIS sensing array.
The array response of the sensor array is defined as 
$\qb_{s}(\theta, \psi)
\triangleq \frac{1}{\sqrt{N_s}}\qb_{s,h}(\theta, \psi)
\otimes \qb_{s,v}(\theta, \psi)\in\mathbb{C}^{N_s\times1}$, 
with $[\qb_{s,h}(\theta, \psi)]_{n_h}
\triangleq e^{-j\pi (n_h-1)\sin\theta\sin\psi}$, while $[\qb_{s,v}(\theta, \psi)]_{n_v}
\triangleq e^{-j\pi (n_v-1)\cos\psi}$. Moreover,
$\qa_{\text{RT}}\!\bigl(\theta_{\text{D}}^{\text{RT}}, \psi_{\text{D}}^{\text{RT}}\bigr)\in\mathbb{C}^{N\times1}$ 
denotes the reflection-side steering vector of the STAR-RIS, where $\theta_{\text{D}}^{\text{RT}}$ and $\psi_{\text{D}}^{\text{RT}}$ 
are the azimuth and elevation departure angles from the STAR-RIS toward the target. 
Specifically, it is defined as 
$\qa_{\text{RT}}(\theta_{\text{D}}^{\text{RT}}, \psi_{\text{D}}^{\text{RT}})
\triangleq \frac{1}{\sqrt{N}}\qa_{\text{RT},h}(\theta_{\text{D}}^{\text{RT}}, \psi_{\text{D}}^{\text{RT}})
\otimes \qa_{\text{RT},v}(\theta_{\text{D}}^{\text{RT}}, \psi_{\text{D}}^{\text{RT}})$, 
with $[\qa_{\text{RT},h}(\theta_{\text{D}}^{\text{RT}}, \psi_{\text{D}}^{\text{RT}})]_{n_h}
= e^{-j\pi (n_h-1)\sin\theta_{\text{D}}^{\text{RT}}\sin\psi_{\text{D}}^{\text{RT}}}$ 
and $[\qa_{\text{RT},v}(\theta_{\text{D}}^{\text{RT}}, \psi_{\text{D}}^{\text{RT}})]_{n_v}
= e^{-j\pi (n_v-1)\cos\psi_{\text{D}}^{\text{RT}}}$.
The sensing signal round trip path loss is calculated as\cite{Yu:Active-RIS-ISAC:TCOM:2024},\cite{Shao:IRS-Sensing:JSAC:2022}
\begin{equation}\label{eq:betaS}
    \beta_s = \sqrt{\frac{\lambda^{2}\Lambda}{(4\pi)^{3} d_s^{4}}},
\end{equation}
where $\Lambda$ is the radar cross section (RCS), while $d_s$ is the distance 
between  the STAR-RIS and the sensing target. 

\begin{remark}
    The proposed sensing framework can be naturally extended to multipath and cluttered propagation environments by incorporating additional scattering, reflection, and cluttered components into the received sensing signal model\cite{wang2024received}. While the current formulation focuses on the dominant target echo, each NLoS path or cluttered component can be characterized by its own path gain, propagation delay, Doppler shift, and angular parameters. Consequently, the received sensing signal can be modeled as the superposition of the desired target echo, multi-path reflections, and clutter-induced interference, enabling the proposed framework to accommodate more realistic sensing scenarios.
\end{remark}

\vspace{-0.4em}
\section{Transmission Protocol: Uplink Training, Channel Estimation and Precoding Design}\label{sec:Transmission Protocol}
In this section, we introduce a transmission protocol framework that enables simultaneous localization, communication, and WPT. Building on this framework, we propose a tensor-based channel estimation algorithm to estimate the individual channels. 
\vspace{-1em}
\subsection{Transmit Signal Model}
The unified transmit signal generated at the BS to enable simultaneous localization,  communication and energy transfer at time slot $t$ is given by
\vspace{-0.2em}
\begin{align}~\label{eq:signal}
  \qx(t) \!= \!&\sum\nolimits_{k_c\in\Kc}\!\sqrt{\pckc} \wkc x^{\mathtt{c}}_{k_c}(t) + \sum\nolimits_{k_p\in\Kp}\!\!
  \sqrt{\ppkp}\wkp x^{\mathtt{p}}_{k_p}(t) \notag\\&+\sqrt{p_{s}}\qw_{s} x_s(t),  
\end{align}
where $\pckc$,  $\ppkp$ and $p_s$ denote the transmit powers allocated to CU, EU, and the sensing signal, respectively; $\qw^{\mathtt{c}}_{k_c}\in\mathbb{C}^{M\times 1}$ and $\qw^{\mathtt{p}}_{k_p}\in\mathbb{C}^{M\times 1}$ are the beamforming vectors associated with CU $k_c$ and EU $k_p$, respectively, with $\Ex\{\Vert \qw^{\mathtt{c}}_{k_c}\Vert^2\}=\Ex\{\Vert \qw^{\mathtt{p}}_{k_p} \Vert^2\}=1$; $x^{\mathtt{c}}_{k_c}(t)$ and $x^{\mathtt{p}}_{k_p}(t)$ are the communication and energy signals for CU $k_c$ and EU $k_p$, respectively, satisfying $\Ex\{x^{\mathtt{c}}_{k_c}(t)x^{\mathtt{c}^H}_{k_c}(t) \}=1$ and $\Ex\{x^{\mathtt{p}}_{k_p}(t)x^{\mathtt{P}^H}_{k_p}(t)\}=1$; $x_s(t)$ is the unit sensing signal, which is a FMCW signal as \cite{He:SLAM:SIGPROC:2023}
\begin{equation}
    x_s(t) =   e^{j \pi \frac{B_s}{T_s} t^2 }\label{st},
\end{equation}
where  \( B_s \) is the frequency sweep bandwidth, and \( T_s \) is the chirp duration.
The signal received by CU $k_c$  is given by
\begin{align} \label{eq:ykc}
    &\qy_{k_c}(t)=\sqrt{\pckc}\qh_{k_c}^T\Thetat\qG \qw^{\mathtt{c}}_{k_c} x^{\mathtt{c}}_{k_c}(t)\notag\\ &+\qh_{k_c}^T\Thetat\qG\sum\nolimits_{i\in\Kc\setminus k_c}\sqrt{\pci}\qw^{\mathtt{c}}_i x^{\mathtt{c}}_i(t) \nonumber\\
    &+ \qh_{k_c}^T\Thetat\qG\sum\nolimits_{k \in\Kp}\sqrt{p^{\mathtt{p}}_k}\wpk x^{\mathtt{p}}_k(t) \nonumber + \sqrt{p_s}\qh_{k_c}^T\Thetat\qG\qw_s x_s(t)\!\\
   & +\qh_{k_c}^T\boldsymbol{\Upsilon}_{t} \qv + n(t),
\end{align}
 where $\qv \sim \mathcal{CN}(\boldsymbol{0}, \sigma_v^2\qI_{N})$ denotes the noise generated by active  elements\cite{Long:Active-RIS:TWC:2021}, while $n(t) \sim \mathcal{CN}(0, \sigma^2)$ is the additive white Gaussian noise (AWGN) at CU $k_c$.

The  signal received at EU $k_p$ is given by
\vspace{-0.2em}
\begin{align}\label{recived signal eu}
y_{k_p}(t) &= \sum\nolimits_{i\in\Kc} \qh_{k_p}^T \mathbf{\Theta}_t \qG\sqrt{\pci} \qw^{\mathtt{c}}_i x^{\mathtt{c}}_i(t) 
+ \qh_{k_p}^T (\mathbf{A} \circ \mathbf{\Theta}_t) \qv
\nonumber\\
&
+ \sum\nolimits_{k\in\Kp} {\qh}_{k_p}^T \mathbf{\Theta}_t {\qG} \sqrt{p^{\mathtt{p}}_k}\qw^{\mathtt{p}}_k x^{\mathtt{p}}_k(t) \notag\\
&+\qh_{k_p}^T \mathbf{\Theta}_t \qG \sqrt{p_s}\qw_s x_s(t)+ n(t).
\end{align}

The  echo signal received by the sensors at the STAR-RIS at time slot $t$ is 
\vspace{-0.2em}
\begin{align}
    \qy_s(t)&=\sqrt{p_{s}}\qH_s\qw_{s} x_s(t)+\sum\nolimits_{k_c\in\Kc}\!\sqrt{\pckc}\qH_s \wkc x^{\mathtt{c}}_{k_c}(t) \notag \\
    &\quad +\sum\nolimits_{k_p\in\Kp}\!\!
  \sqrt{\ppkp}\qH_s\wkp x^{\mathtt{p}}_{k_p}(t) + \qn_s(t)\notag \\
  &=\qH_s \qs(t)+ \tilde{\qn}_s(t),
\end{align}
where $\qs(t)=\sqrt{p_{s}}\qw_{s} x_s(t)$; $\qn_s(t)$ is the AWGN at the sensors, while
\vspace{0.2em}
\begin{align}   \tilde{\qn}_s(t)&\triangleq\sum\nolimits_{k_c\in\Kc}\!\sqrt{\pckc}\qH_s \wkc x^{\mathtt{c}}_{k_c}(t)\nonumber\\
&+\sum\nolimits_{k_p\in\Kp}\!\!
  \sqrt{\ppkp}\qH_s\wkp x^{\mathtt{p}}_{k_p}(t) + \qn_s(t), 
\end{align}
represents the equivalent sensing noise, which consists of $\qn_s(t)$ and the interference caused by the communication and energy signals.

\vspace{-1.5em}
\subsection{Channel Estimation via Uplink Training}
We consider an UL frame structure comprising \( P \) consecutive time slots dedicated to UL training, during which the STAR-RIS applies a distinct fully passive phase configuration \(\bm{\phi}_{t,p}^{\mathrm{ul}} \) in each slot, for \( p = 1, \dots, P \). 
Let $K_u\triangleq K_c+K_p$ users simultaneously transmit UL pilot signals to the BS. The pilot sequences assigned to the users are assumed to be mutually orthogonal; hence, the \( K_u \times P \) pilot matrix \( \mathbf{X}_P \) satisfies \( \mathbf{X}_P\mathbf{X}_P^H  = \mathbf{I}_{K_u} \), where \( P \geq K_u \) is the pilot length. It is assumed that all users transmit the same pilot sequences in each of the $P$ training time slots. Therefore, the pilot signal received at the BS during time slot $p$ is given by
\vspace{0.3em}
\begin{equation}
    \mathbf{Y}_p=\qG^T\operatorname{diag}(\bm{\phi}_{t,p}^{\mathrm{ul}})\qH^T\mathbf{X}_P +\mathbf{N}_p,
\end{equation}
where $\mathbf{N}_p \in \mathbb{C}^{M \times P}$ denotes the AWGN matrix at the BS during time slot $p$, whose entries follow the complex Gaussian distribution $\mathcal{CN}(0, \sigma^2)$.
By removing the pilot symbols after right-multiplying $\mathbf{Y}_p$ with $\mathbf{X}_P^H$, we obtain
\vspace{0.3em}
\begin{equation}
    \mathbf{Z}_p = \qG^T \operatorname{diag}(\bm{\phi}_{t,p}^{\mathrm{ul}}) \qH^T + \tilde{\mathbf{N}},
\end{equation}
where $\tilde{\mathbf{N}} = \mathbf{N}_p \mathbf{X}_P^H \in \mathbb{C}^{M \times K_u}$ denotes the equivalent noise matrix.
After receiving the signals across all $P$ time slots, we can construct a three-dimensional tensor $\tens{Z} \in \mathbb{C}^{M \times K_u \times P}$ by
\vspace{0.2em}
\begin{equation}
    \tens{Z}(:,:,p) = \mathbf{Z}_p, \quad p = 1, 2, \dots, P.
\end{equation}

Let $\qG^T = [\qg_1, \qg_2, \dots, \qg_N] \in \mathbb{C}^{M \times N}$ and $\qH^T = [\qh_1, \qh_2, \dots, \qh_{(K_c+K_p)}]^T \in \mathbb{C}^{N \times (K_c+K_p)}$, where $\qg_n \in \mathbb{C}^{M\times1}$ and $\qh_n \in \mathbb{C}^{K_u\times1}$, denote the $n$th spatial and user-related components, respectively. Let $\boldsymbol{\phi}_{t,n}^{\mathrm{ul}} = [\phi^{\mathrm{ul}}_{1}, \phi^{\mathrm{ul}}_{2}, \dots, \phi^{\mathrm{ul}}_{N}]^T \in \mathbb{C}^{N\times1}$ represent the STAR-RIS phase configuration at time slot $n$. We define $\boldsymbol{\phi}^{\mathrm{ul}}_{t,p} = [\phi^{\mathrm{ul}}_{1}, \phi^{\mathrm{ul}}_{2}, \dots, \phi^{\mathrm{ul}}_{P}]^T \in \mathbb{C}^{P\times1}$ as  the $p$-th STAR-RIS element over all $P$ slots, and collect them as $\bm{\Phi}^{\mathrm{ul}}_t = [\bm{\phi}^{\mathrm{ul}}_{t,1}, \bm{\phi}^{\mathrm{ul}}_{t,2}, \dots, \bm{\phi}^{\mathrm{ul}}_{t,N}] \in \mathbb{C}^{P \times N}$.
With these definitions, the received tensor $\tens{Z} \in \mathbb{C}^{M \times K_u \times P}$ can be represented via a rank-$N$ PARAFAC decomposition as
\begin{equation}
    \tens{Z} = \sum\nolimits_{n=1}^{N} \qg_n \circ \qh_n \circ \bm{\phi}^{\mathrm{ul}}_{t,n} + \mathcal{N},
\end{equation}
where $\mathcal{N} \in \mathbb{C}^{M \times K_u \times P}$ denotes the additive noise tensor. The SNR of the received tensor signal is defined as $\mathrm{SNR} =\frac{\left\|\tens{Z}_0\right\|_F^2}
{\left\|\mathcal{N}\right\|_F^2}.$ And $\tens{Z}_0$ denotes the noise-free received tensor signal. 
By extracting the horizontal, lateral, and frontal slices of the tensor $\tens{Z} \in \mathbb{C}^{M \times K_u \times P}$, corresponding to fixing the first, second, and third dimensions, respectively, we obtain the mode-1, mode-2, and mode-3 unfoldings of the tensor as:
\vspace{0.2em}
\begin{subequations}
   \begin{align}
    \mathbf{Z}_{(1)} &= \qG (\qH \odot \bm{\Phi}^{\mathrm{ul}}_t)^T + \mathbf{N}_{(1)} , \\
    \mathbf{Z}_{(2)} &= \qH (\qG \odot \bm{\Phi}^{\mathrm{ul}}_t)^T + \mathbf{N}_{(2)}, \\
    \mathbf{Z}_{(3)} &= \bm{\Phi}^{\mathrm{ul}}_t (\qG \odot \qH)^T + \mathbf{N}_{(3)},
\end{align} 
\end{subequations}
where $\mathbf{Z}_{(1)} \in \mathbb{C}^{M \times K_u P}$, 
$\mathbf{Z}_{(2)} \in \mathbb{C}^{K_u \times MP}$, and 
$\mathbf{Z}_{(3)} \in \mathbb{C}^{P \times M K_u}$;  $\mathbf{N}_{(1)}$, $\mathbf{N}_{(2)}$, and $\mathbf{N}_{(3)}$ denote the mode-1, mode-2, and mode-3 unfoldings of the noise tensor $\mathcal{N} \in \mathbb{C}^{M \times K_u \times P}$, respectively, corresponding to the additive noise in the horizontal, lateral, and frontal slices.

 We build upon the PARAFAC-based modeling in STAR-RIS-assisted signal processing and propose an iterative channel estimation algorithm based on ALS. The channel factors are estimated by minimizing the following Frobenius norm:
\begin{equation}
    \min_{\hat{\qG},\, \hat{\qH}} 
    \Big\| \tens{Z} - \sum\nolimits_{n=1}^{N} \hat{\qg}_n \circ \hat{\qh}_n \circ {\bm{\phi}}^\text{ul}_{t,n} \Big\|_F^2,
\end{equation}
where $\tens{Z} $ denotes the observed tensor, while $\hat{\qg}_n$ and $\hat{\qh}_n$ are the estimated rank-one components. The detailed procedure is provided in \textbf{Algorithm~\ref{alg:parafac}.}

\begin{algorithm}[t]
\caption{PARAFAC-ALS-based channel estimation}
\label{alg:parafac}
\begin{algorithmic}[1]  
\Input Received tensor $\tens{Z} \in \mathbb{C}^{K_u \times M \times P}$, STAR-RIS UL training phase matrix $\bm{\Phi}_t^{\mathrm{ul}} \in \mathbb{C}^{P \times N}$, max number of iterations $i_\mathrm{max}$, tolerance $\epsilon$.
\Output Estimated channels $\hat{\qG} \in \mathbb{C}^{N \times M}$, $\hat{\qH} \in \mathbb{C}^{K_u \times N}.$

\State Initialize $\hat{\qH}^{(1)}$ and $\hat{\qG}^{(1)}$
\For{$i = 2$ to $i_\mathrm{max}$}
    \State Compute Khatri-Rao product: $\mathbf{A}_2 = \bm{\Phi}_t^{\mathrm{ul}} \odot \hat{\qH}^{(i-1)}$;
    \State Update $\hat{\qG}^{(i)} = \left( \mathbf{A}_2^H \mathbf{A}_2 \right)^{-1} \mathbf{A}_2^H \mathbf{Z}_{(2)}$;
    \State Normalize each row of $\hat{\qG}^{(i)}$ by its first entry to remove scaling ambiguity;
    \State Compute Khatri-Rao product: $\mathbf{A}_1 = \bm{\Phi}_t^{\mathrm{ul}} \odot \left( \hat{\qG}^{(i)} \right)^T$
    \State Update $\hat{\qH}^{(i)} = \left( \mathbf{A}_1^H \mathbf{A}_1 \right)^{-1} \mathbf{A}_1^H \mathbf{Z}_{(1)}^T$
    \State Compute the fitting error $e_i$ as the Frobenius norm of the reconstruction residual: $e_i = \left\| \tens{Z} - \sum\nolimits_{n=1}^{N} \qg_n \circ \qh_n \circ \bm{\phi}^{\mathrm{ul}}_{t,n} \right\|_F^2$
    \State Compute relative error: $\delta_i = \frac{e_{(i-1)} - e_i}{e_i}$
    \If{$|\delta_i| < \epsilon$}
        \State \textbf{break}
    \EndIf
\EndFor
\State \Return Final estimates $\hat{\qG} = \hat{\qG}^{(i)}$, $\hat{\qH} = \hat{\qH}^{(i)}$.
\end{algorithmic}
\end{algorithm}
\setlength{\textfloatsep}{0.25cm}

\vspace{-0.5em}
\subsection{Precoding Design}
In the proposed downlink (DL) transmission, distinct precoding strategies are adopted for CUs, EUs, and the sensing target. For CUs, zero-forcing (ZF) precoding is used to mitigate multi-user interference, while for EUs, maximum-ratio transmission (MRT) is used  to maximize the received power~\cite{MohammadiRIS_SWIPT_opt:TCOM:2025}. Accordingly, we have
\vspace{0.3em}
\begin{subequations}
    \begin{align}
    \qw^{\mathtt{p}}_{k} &=\alpha_{k}^{\mathrm{MRT}}\hat{\qf}_k^{H},\quad k \in\mathcal{K}_p, \\
     \qw^{\mathtt{c}}_{k} &=\alpha_{k}^{\mathrm{ZF}} \hat{\qH}^H \big( \hat{\qH}\hat{\qH}^H \big)^{-1}\qe_{k}, \quad k\in \mathcal{K}_c,
\end{align}
\end{subequations}
where $\hat{\qf}_k\triangleq\hat{\qh}_{k}^{T}\mathbf{\Theta}_t\hat{\qG}$; $\alpha_{k}^{\mathrm{MRT}}=1/\|\hat{\qf}_k^{H}\|$; $\hat{\qH}=[\hat{\qf}_k: \forall k\in\mathcal{K}_c]$; $\alpha_{k}^{\mathrm{ZF}}=\sqrt{\tr(( \hat{\qH}\hat{\qH}^H )^{-1})}$; $\qe_k$ is the $k$-th column of  $\qI_{M}$.

For target sensing, CSI is generally unavailable due to the unknown target location; thus, we employ a power-maximizing approach by maximizing the signal power impinging on the STAR-RIS, formulated as a Rayleigh quotient problem. The solution corresponds to the right singular vector of $\mathbf{\hat{G}}$ associated with its largest singular value, as
\vspace{0.2em}
\begin{equation}
  \qw_s = \frac{\qv_1}{\|\qv_1\|}, 
\end{equation}
where $\mathbf{\hat{G}} = \mathbf{U}\boldsymbol{\Sigma}\qV^H$ and $\qv_1 =[\qV]_{:,1}$. The resulting precoders are subsequently utilized in the optimization framework. In addition, the STAR-RIS further enhances sensing performance through the optimization of its phase-shift coefficients, as detailed in Section V. 

\vspace{-0.5em}
\section{Performance Metrics} \label{sec:Performance}
To comprehensively evaluate the proposed multifunctional framework, several performance metrics are considered, each corresponding to a different system functionality. Specifically, the SE and communication SINR characterize communication performance, while the sensing SINR and CRB quantify sensing reliability and localization accuracy. In addition, the harvested energy measures the effectiveness of WPT, whereas the EE captures the tradeoff between communication throughput and total power consumption. It is important to note that these metrics are coupled through the STAR-RIS coefficients, active/passive mode scheduling, beamforming design, and transmit power allocation. As a result, improving one functionality may positively/negatively affect the performance of the others. Therefore, a joint analysis of these metrics is necessary to characterize the communication--sensing--powering tradeoffs and to design an energy-efficient multifunctional STAR-RIS-assisted system. 

\vspace{-0.7em}
\subsection{Localization: Performance Evaluation and Approach}
\subsubsection{Sensing SINR}
The quality of target  localization is evaluated through the sensing SINR, which characterizes the strength of the useful echo relative to the interference and noise\cite{Wang:TWC:2023}. Specifically, the sensing SINR is expressed as \eqref{SINR_s} on the top of next page, where the numerator represents the power of the echo signal reflected by the target and captured by the STAR-RIS sensors. The denominator accounts for several sources of impairments, including the interference from communication and power transfer signals, the additional amplification noise introduced by active STAR-RIS elements, and the thermal noise at the sensors. This metric directly reflects the reliability of sensing and localization, and will serve as a key constraint in the subsequent resource allocation and optimization framework to ensure robust sensing performance.
\begin{figure*}[t]  
    \begin{equation}\label{SINR_s}
    \mathrm{SINR}_{s}=
\frac{p_s\Vert\qH_s\qw_s\Vert^2}{\sum_{i \in\Kc}p^{\mathtt{c}}_i\Vert\qH_s \qw^{\mathtt{c}}_{i}\Vert^2 +\sum_{k \in\Kp}p^{\mathtt{p}}_k\Vert\qH_s\wpk\Vert^2 + \beta_s^2\sigma_v^2\Vert\qb_s\qa^T_{\text{RT}}\mathbf{A} \circ \mathbf{\Theta}_{r}\Vert^2 + \sigma^2}
    \end{equation}
            \vspace{-1.3em}
        \hrulefill  
\end{figure*}

\vspace{-0.6em}
\subsubsection{CRB Derivation}
We evaluate the sensing performance of the proposed hybrid STAR-RIS-assisted system by deriving the CRB for the estimation of the target's 2D AoD parameters.
The sensors collect the echo signals across $T$ sensing snapshots, with each snapshot representing one observation of the received echo signal, expressed as
\begin{equation}\label{ys_T}
    \qY_s=\qH_s \qS + \tilde{\qN}_s,
\end{equation}
where
\begin{align*}
    \qY_s&=[ \qy_s(1),\qy_s(2),\dots\qy_s(T)]\in\mathbb{C}^{N_s \times T},\\
    \qS&=[s(1),s(2),\dots,s(T)]\in\mathbb{C}^{M\times T},\\
    \tilde{\qN}_s&=[\tilde{\qn}_s(1),\tilde{\qn}_s(2),\dots,\tilde{\qn}_s(T)]\in\mathbb{C}^{N_s \times T}.
\end{align*}
The corresponding time delay due to the distance between the target and  sensors is  $\tau$. The received echo signal is $s(t-\tau)=e^{j(2\pi f_0 (t-\tau) + \pi \frac{B_s}{T_s} (t-\tau)^2)}$ where \( f_0 \) is the starting carrier frequency.
Define the unknown parameter vector as
\begin{equation}
    \boldsymbol{\zeta} =
    \begin{bmatrix}
    \theta & \psi & \tau&\Re\{\beta_s\} & \Im\{\beta_s\}
    \end{bmatrix}^T \in \mathbb{R}^{5\times 1}.
\end{equation}
Let $\beta_s = \beta_R + j\beta_I$, where 
$\beta_R = \Re\{\beta_s\},~ \beta_I = \Im\{\beta_s\}$.  Thus, the Fisher Information
Matrix (FIM) is given by
\vspace{0.2em}
\begin{equation}
\mathbf{J}(\boldsymbol{\zeta}) =
\begin{bmatrix}
J_{\theta\theta} & J_{\theta\psi} & J_{\theta\tau} & J_{\theta\beta_R} & J_{\theta\beta_I} \\
J_{\psi\theta} & J_{\psi\psi} & J_{\psi\tau} & J_{\psi\beta_R} & J_{\psi\beta_I} \\
J_{\tau\theta} & J_{\tau\psi} & J_{\tau\tau} & J_{\tau\beta_R} & J_{\tau\beta_I} \\
J_{\beta_R\theta} & J_{\beta_R\psi} & J_{\beta_R\tau} & J_{\beta_R\beta_R} & J_{\beta_R\beta_I} \\
J_{\beta_I\theta} & J_{\beta_I\psi} & J_{\beta_I\tau} & J_{\beta_I\beta_R} & J_{\beta_I\beta_I}
\end{bmatrix}.
\end{equation}
\begin{proposition}
Accordingly, the CRB of estimating target location can be expressed as
\vspace{0.2em}
\begin{equation}
\mathrm{CRB}=\tr(\mathbf{J}^{-1}(\boldsymbol{\zeta})).
\end{equation}
\end{proposition}
\begin{proof}
    The  detailed expressions of the FIM elements are derived in Appendix \ref{pro_CRB}. 
\end{proof}

\subsubsection{Target Location Estimation}
We estimate the target location by combining the angular information from the MUSIC algorithm with the range information extracted from the FMCW waveform\cite{Ni:Hybrid-STAR-RIS:SPAWC:2025}. The echo signals at the STAR-RIS sensors form a covariance matrix, whose eigen decomposition yields the noise subspace $\mathbf{E}_n$. The MUSIC spectrum is then computed as
\vspace{0.2em}
\begin{equation}
    P_{\text{MUSIC}}(\theta,\psi)
    = \frac{1}{\left| \mathbf{b}_{s}^H(\theta,\psi)\,
    \mathbf{E}_n \mathbf{E}_n^H\,
    \mathbf{b}_{s}(\theta,\psi) \right|},
\end{equation}
while the target angles are obtained by
\vspace{0.3em}
\begin{equation}
 (\hat{\theta}, \hat{\psi})
 = \arg\max_{\theta,\psi} P_{\text{MUSIC}}(\theta,\psi).
\end{equation}
For range estimation, the FMCW waveform produces a beat signal whose frequency encodes the round trip time delay. The received beat signal can be written as
\vspace{0.2em}
\begin{equation}
\qs_{b}(t)
=\qH_s\sqrt{p_s}\qw_s
e^{-j2\pi\left(f_0 \tau + \frac{B_s}{T_s} t \tau - \frac{B_s}{2T_s} \tau^{2}\right)}.
\end{equation}
By applying a Fast Fourier Transform  to the received FMCW beat signal, the beat frequency $f_b$ is estimated and mapped to the round-trip delay as $f_b = \frac{B_s}{T_s}\tau$. The corresponding delay $\hat{\tau} = \frac{T_s}{B_s}f_b$ yields the target distance estimate $\hat{D}_t = \frac{c\hat{\tau} - d_{BR}}{2}$. The estimated Cartesian coordinates of the target are then given by $(\hat{x},\, \hat{y},\, \hat{z})$, where
$\hat{x} = \hat{D}_t \cos\hat{\theta}\cos\hat{\psi}$,
$\hat{y} = \hat{D}_t \cos\hat{\theta}\sin\hat{\psi}$, 
and 
$\hat{z} = \hat{D}_t \sin\hat{\theta}$.
 To assess the accuracy of the proposed localization method, we adopt the root mean square error (RMSE) as the evaluation metric, which is expressed as
\vspace{0.2em}
\begin{equation}
\mathrm{RMSE}
= \sqrt{
(\,x - \hat{x}\,)^{2}
+ (\,y - \hat{y}\,)^{2}
+ (\,z - \hat{z}\,)^{2}
}.
\end{equation}

\begin{remark}
    The proposed localization approach can be extended to multi-target scenarios by extracting multiple angular peaks from the MUSIC spectrum and multiple  beat frequency components from the FMCW signal.
\end{remark}

\vspace{-1em}
\subsection{SE and EE Performance Metrics}
By using the received signal in~\eqref{eq:ykc}, the SE of CU $k_c$ is adopted as the performance metric and is given by
\begin{equation}\label{R_kc}
  \mathrm{SE}_{k_c} = \log_2 \left( 1 + \mathrm{SINR}_{k_c} \right),
\end{equation}
where $\mathrm{SINR}_{k_c}$ denotes the received SINR at CU $k_c$, given in ~\eqref{SINR_c}, on top of next page.

\begin{figure*}[t]  
    \begin{align}\label{SINR_c}
    &\mathrm{SINR}_{k_c}=
\frac{\pckc|\qh_{k_c}^T\Thetat\qG\qw^{\mathtt{c}}_{k_c}|^2}{\sum_{i \in\Kc\setminus k_c}p^{\mathtt{c}}_i|\qh_{k_c}^T\Thetat \qG \qw^{\mathtt{c}}_{i}|^2 +\sum_{k \in\Kp}p^{\mathtt{p}}_k|\qh_{k_c}^T\Thetat\qG\wpk|^2 +p_s|\qh_{k_c}^T\Thetat\qG\qw_s|^2+ \Vert\qh_{k_c}^T\mathbf{A} \circ \mathbf{\Theta}_{t}\Vert^2\sigma_v^2 + \sigma^2}
    \end{align}
        \hrulefill  
        \vspace{-0.7em}
\end{figure*}
The EE is a key performance metric that  captures the trade-off between the achieved communication throughput and the total power consumption of the system, which is defined as
\vspace{0.2em}
\begin{equation}
    \eta_{\mathrm{E}} = \frac{B\sum\nolimits_{k_c\in\Kc} \mathrm{SE}_{k_c} }{P_t + N  (P_c + P_b) +  N_\text{act}P_a},
\end{equation}
where $B$ denotes the system bandwidth; $\mathrm{SE}_{k_c}$ is the SE of CU $k_c$ given by \eqref{R_kc}; $P_t$ represents the transmit power at the BS; $P_c$ and $P_b$ are the static circuit and bias powers per STAR-RIS element; $P_a$ is the additional power consumed by active elements. Moreover, we define  $P_r$ as the dynamic power for amplifying the incident signal and the noise, given by
\vspace{0.2em}
\begin{align}
        P_r=&\sum\nolimits_{k \in \mathcal{K}}p_k(\Vert\boldsymbol{\Upsilon}_t\qG\qw_k\Vert^2+\Vert\boldsymbol{\Upsilon}_r\qG\qw_k\Vert^2)\notag\\&+(\Vert\boldsymbol{\Upsilon}_r\Vert^2+\Vert\boldsymbol{\Upsilon}_t\Vert^2)\sigma_v^2,
\end{align}
where $\mathcal{K}=\Kc\cup\Kp\cup{s}$ contains the indices of all CUs and EUs as well as sensing target.
It is worth noting that the terms $P_r$ and $P_a$ represent different components of the power consumption of active STAR-RIS elements. Specifically, $P_a$ denotes the fixed circuit power required to activate the active elements,  while  $P_r$  depends on the strength of the incoming signals. 
\vspace{-0.7em}
 \subsection{Harvested Energy}
For each EU $k_p$, the received signal power consists of the communication signals, the dedicated power transfer signals, the sensing waveform, as well as the additional noise from active STAR-RIS elements.  Accordingly, the harvested input energy can be expressed as
\vspace{0.3em}
\begin{align}
Q_{k_p} &\!= \!\! \sum\nolimits_{i\in\Kc} \pci \big| \qh_{k_p}^T \mathbf{\Theta}_t \qG \qw^{\mathtt{c}}_i \big|^2
\!+ \!\sum\nolimits_{k\in\Kp} p^{\mathtt{p}}_k\big| \qh_{k_p}^T\Thetat\qG\wpk \big|^2\nonumber\\
&+ p_s\big| \qh_{k_p}^T\Thetat\qG\qw_s \big|^2 + \big\Vert \qh_{k_p}^T (\mathbf{A} \circ \mathbf{\Theta}_t) \big\Vert^2 \sigma_v^2 .\label{Q_kp}
\end{align}
To account for the nonlinear characteristics of practical radio-frequency (RF) energy harvesting circuits, we adopt the widely used sigmoidal nonlinear EH model. The harvested  power at EU $k_p$ is\cite{Mumtaz:EH-CF-MIMO:TGCN:2025}
\vspace{0.2em}
\begin{equation}
P^{NL}_{k_p}(Q_{k_p}) = \frac{\Psi(Q_{k_p}) - \phi\,\Omega}{1 - \Omega},
\end{equation}
where $\phi$ denotes the maximum harvested energy, $\xi$ and $\chi$ characterize the sharpness and threshold of the EH circuit response, and $\Omega = \frac{1}{1 + \exp(\xi \chi)}$ is a normalization factor adjusting EH behavior under low input power. The nonlinear energy mapping is given by
\vspace{0.2em}
\begin{equation}\label{Psi_Qk}
    \Psi(Q_{k_p}) = \phi\left(1 + \exp\left(-\xi \left(Q_{k_p} - \chi\right)\right)\right)^{-1}.
\end{equation}

\section{Energy Efficiency Maximization}\label{sec: EE optimization}
In this section, we formulate and solve the EE maximization problem for the proposed hybrid STAR-RIS-assisted system. The channel estimates are used for precoder design and significantly affect the optimization performance, since the sum SE, harvested energy, and sensing SINR depend on the resulting effective channels.  The optimization objective is to maximize the EE, subject to communication QoS, sensing reliability, harvested energy, and hardware-power constraints. 
The optimization problem can be formulated as
\vspace{0.3em}
\begin{subequations}\label{opt_eq_O}
\begin{align}
&\max_{\mathbf{\Theta}_t,\mathbf{\Theta}_r,\qp,\qA}\;\eta_{\mathrm{E}}\triangleq\frac{\mathrm{SE}_{\mathrm{sum}}(\mathbf{\Theta}_t,\mathbf{A},\qp)}{P_{\mathrm{tot}}(\mathbf{\Theta}_t,\mathbf{A},\qp)} \\
\mathrm{s.t.}\quad &\qp \succeq \mathbf{0} , \ \ \sum\nolimits_{k \in \mathcal{K}} p_k\leq P_t^\mathrm{max},   \label{BS_Pt} \\
&\mathrm{SE}_{k}(\mathbf{\Theta}_t,\mathbf{A},\qp) \ge \mathrm{SE}_{k}^\mathrm{min},\quad \forall k\in\mathcal{K}_c\label{Rkc_QOS},\\
 &P^{\mathrm{NL}}_{k}(\mathbf{\Theta}_t,\mathbf{A},\qp)\ge \Gamma_E,\quad \forall k\in\mathcal{K}_p,\label{constraint-P^NL}\\ 
&\sum\nolimits_{k \in \mathcal{K}}p_k\Vert\boldsymbol{\Upsilon}_t\qG\qw_k\Vert^2+\sum\nolimits_{k \in \mathcal{K}}p_k\Vert\boldsymbol{\Upsilon}_r\qG\qw_k\Vert^2\notag\\&+\Vert\boldsymbol{\Upsilon}_r\Vert^2\sigma_v^2+\Vert\boldsymbol{\Upsilon}_t\Vert^2\sigma_v^2\leq P_r^{\mathrm{max}},\label{Qos_Pr}\\
&\SINR_s(\mathbf{\Theta}_r,\mathbf{A},\qp)\geq \gamma_s^{\mathrm{min}}, \label{QOS_SSINR}\\
&\theta_n \in(0,2\pi],\label{RIS_law1}\\
&1<\beta_{t,n},\beta_{r,n}\leq\rho_\mathrm{max}, 
\forall n\in\mathcal{A},\\
&\beta_{t,n}^2+\beta_{r,n}^2\leq1, n\notin\mathcal{A},\\
&0<\beta_{r,n}<1,n\notin\mathcal{A},\\
&0<\beta_{t,n}<1,n\notin\mathcal{A},\label{RIS_law2}\\
&\bm{\alpha}_n\in\{0,1\},\label{alpha_selection}
\end{align}
\end{subequations}
where $\qp \! = \![\,p^{\mathtt{c}}_1,\ldots,p^{\mathtt{c}}_{K_c},\,p^{\mathtt{p}}_1,\ldots,p^{\mathtt{p}}_{K_p},\,p_s\,]^T\in \mathbb{R}^{(\!K_c \!+\! K_p \!+\! 1)\times 1}$ stores the power allocation coefficients; $\mathrm{SE}_{\mathrm{sum}}=\sum\nolimits_{k\in\mathcal{K}_c}\mathrm{SE}_{k}$. The constraint \eqref{BS_Pt} originates from the total power budget at the BS and ensures that the sum of all transmit powers does not exceed this budget, $ P_t^\mathrm{max}$. 
The constraints in \eqref{Rkc_QOS} and \eqref{constraint-P^NL} are imposed to guarantee the QoS requirements of the CUs and EUs, respectively. Specifically, \eqref{Rkc_QOS} ensures a minimum SE of $\mathrm{SE}_{k}^{\min}$ for each CU, while \eqref{constraint-P^NL} guarantees a minimum harvested energy level of $\Gamma_E$ for each EU. It is worth noting that the WPT constraint in \eqref{constraint-P^NL} is not a simple linear received-power constraint. The harvested energy depends jointly on the communication signals, sensing waveform, dedicated energy signals, active-element noise, and the nonlinear EH circuit response. Consequently, the WPT requirement is tightly coupled with the power-allocation strategy, STAR-RIS transmission/reflection coefficients, and active/passive mode assignment, thereby influencing both the feasible solution space and the resulting system design. 
Moreover, constraint~\eqref{Qos_Pr} limits the amplification power of the active RIS elements based on their hardware capabilities, where $P_r^{\mathrm{max}}$ denotes the maximum allowable transmit power. Constraint (39f) guarantees sensing reliability by ensuring that the received sensing SINR exceeds the minimum required threshold $\gamma_s^{\mathrm{min}}$. Although localization performance can be characterized by the CRB, directly incorporating a CRB constraint into the EE maximization problem would lead to a highly non-convex formulation due to the dependence of the FIM on multiple coupled design variables. Therefore, consistent with several existing ISAC studies, such as \cite{demirhan2025cell},\cite{sankar2024beamforming} and \cite{ye2025scnr}, that employ tractable sensing-quality metrics during optimization, the sensing SINR is adopted in our proposed framework, while the resulting localization performance is explicitly evaluated through the CRB and RMSE analysis.  Finally, \eqref{RIS_law1}-\eqref{alpha_selection} define the amplitude constraints for active and passive elements based on the power conservation law.

Problem \eqref{opt_eq_O} is a highly nonconvex fractional program with coupled variables, making standard convex tools inapplicable. To ensure tractability, we decompose it into three subproblems: element-mode selection, power allocation under a given STAR-RIS configuration, and phase optimization given the power allocation. Using fractional programming and SCA, efficient iterative algorithms are developed to iteratively obtain high-quality solutions.

\vspace{-1em}
\subsection{Element Mode Selection}
Before  optimizing the RIS amplitude and phase coefficients and the BS transmit power, the operating mode of each RIS element must first be determined. To obtain the optimal element mode selection, exhaustive search provides the optimal approach via  all possible activation sets  $P_{\mathcal{S}}\!=\!\binom{N}{N_{\text{act}}}$. This approach  provides  an upper bound on the achievable performance of active-element selection. However, its computational complexity grows exponentially with the number of RIS elements, which makes it applicable only to scenarios with a relatively small surface size. To overcome this, we propose a low-complexity greedy strategy that incrementally selects active elements by their performance contribution. Though suboptimal, it offers acceptable results with complexity reduced from $\mathcal{O}(2^N)$ to $\mathcal{O}(N N_\text{act})$.

Let $\mathcal{S}$ and $\mathcal{R}$ denote the sets containing the indices of the already selected
active RIS elements and the remaining candidate elements, respectively. As summarized in \textbf{Algorithm~\ref{alg:greedy_active}}, we initialize the algorithm with
$\mathcal{S} = \emptyset$, $\mathcal{R} = \{1,\ldots,N\}$, and a full zero activation vector $\boldsymbol{\alpha}$. At each outer iteration, for every candidate element,  we temporarily form an activation vector that turns on the elements in $\mathcal{S} \cup \{n^*\}$, compute the corresponding EE, and record the best value $\eta_{\text{best}}$.
This greedy procedure is repeated until the number of selected elements
satisfies $|\mathcal{S}| = N_{\text{act}}$, after which the final activation vector
$\boldsymbol{\alpha}$ is obtained by setting $\alpha_n = 1$ for $n \in \mathcal{S}$
and $\alpha_n = 0$ otherwise.

\begin{algorithm}[t]
\caption{Greedy algorithm for active RIS element selection}
\label{alg:greedy_active}
\begin{algorithmic}[1]
\State \textbf{Input:} $N$, $N_{\mathrm{act}}.$
\State \textbf{Initialize:} selected set $\mathcal{S} \gets \emptyset$, remaining set $\mathcal{R} \gets \{1,2,\ldots,N\}$.
\For{$k = 1$ to $N_{\mathrm{act}}$}
    \State $\eta_{\mathrm{best}} \gets -\infty$, $n^{\star} \gets 0$;
    \For{each $n \in \mathcal{R}$}
        \State Construct $\boldsymbol{\alpha} \in \{0,1\}^N$ with
        $\alpha_i = 1$ if $i \in \mathcal{S} \cup \{n\}$ and $\alpha_i = 0$ otherwise;
        \State Given $\boldsymbol{\alpha}$, compute the EE
        $\eta_{\mathrm{E}}(\boldsymbol{\alpha}) = \dfrac{R^{\text{sum}}(\boldsymbol{\alpha})}{P_{\text{tot}}(\boldsymbol{\alpha})}$.
        \If{$\eta_{\mathrm{E}}(\boldsymbol{\alpha}) > \eta_{\mathrm{best}}$}
            \State $\eta_{\mathrm{best}} \gets \eta_{\mathrm{E}}(\boldsymbol{\alpha})$,  $n^{\star} \gets n$;
        \EndIf
    \EndFor
    \State $\mathcal{S} \gets \mathcal{S} \cup \{n^{\star}\}$.
    \State $\mathcal{R} \gets \mathcal{R} \setminus \{n^{\star}\}$.
\EndFor
\State Construct the final selection vector $\boldsymbol{\alpha}^{\star} \in \{0,1\}^N$ with
$\alpha_i^{\star} = 1$ if $i \in \mathcal{S}$ and $\alpha_i^{\star} = 0$ otherwise.
\State \textbf{Output:} greedy active element selection vector $\boldsymbol{\alpha}^{\star}$.
\end{algorithmic}
\end{algorithm}
\setlength{\textfloatsep}{0.3cm}

Although the greedy algorithm significantly reduces the computational complexity compared with the exhaustive search, its performance is inherently sensitive to the initialization. Specifically, a single greedy trajectory may become trapped in a local optimum because the first few selected elements strongly influence all subsequent decisions. To address this issue, we propose a heuristic-enhanced multi-seed greedy \textbf{(HE-MSG)} algorithm. To begin with, several promising initial active-element seeds are heuristically generated by channel-gain selection \textbf{(CGS)} of each RIS element as follows. The contribution of the $n$-th RIS element to the cascaded BS–RIS–user link can be quantified by its individual channel gain as
\begin{equation}
    g_n = \left\| [\hat{\qG}^T]_{:,n} [\hat{\qh}^T]_{n,:} \right\|^2, \quad n = 1,2,\ldots,N.
\end{equation}
We  sort the  channel gains $g_n$ to obtain the seed set $\mathcal{S}_{\mathrm{seed}}$ and each seed initializes a greedy path search. The detailed procedure is summarized in  
\textbf{Algorithm~\ref{alg:hemsg}}.  Compared with random selection and CGS, the proposed HE-MSG strategy jointly exploits channel information and the EE objective during active-element selection, leading to improved and more robust performance. Although the use of multiple search seeds introduces a moderate increase in computational complexity relative to single-seed greedy selection, the overall complexity remains significantly lower than that of exhaustive search while achieving a more favorable performance-complexity tradeoff.

\vspace{-1em}
\subsection{Power Allocation  Subproblem}
To begin with, we first consider the case where the mode scheduling variables $\bm{\alpha}$ and the STAR-RIS transmission coefficients  are fixed.  
Therefore, the first subproblem focuses on optimizing the transmit power allocation vector
$\qp \triangleq \big[ \qp^{\mathtt{c}},\ \qp^{\mathtt{p}},\ p_s \big],$ where $\qp^{\mathtt{c}} \in \mathbb{R}^{K_c \times 1}$ and
$\qp^{\mathtt{p}} \in \mathbb{R}^{K_p \times 1}$ denote the transmit power allocation vectors
for CUs and EUs, respectively, and $p_s \in \mathbb{R}$ denotes the sensing signal power. The resulting optimization problem is formulated as follows:
\vspace{0.2em}
\begin{subequations}\label{P1}
\begin{alignat}{2}
&\ \max_{\qp}      
&\qquad& \eta_{\mathrm{E}} = \frac{\mathrm{SE}_{\mathrm{sum}}(\qp)}{P_{\mathrm{tot}}(\qp)} \label{P1.1a}\\
&\hspace{1em}\text{s.t.}  &      & \qp \succeq \mathbf{0} , \ \ \sum\nolimits_{k \in \mathcal{K}} p_k\leq P_t^\mathrm{max},  \label{P1.1b}\\
&      &      & \SINR_{k}(\qp)\geq \gamma_{k}^{\mathrm{min}}, \forall k \in \mathcal{K}_c, \label{P1constriant_min_sinrkc}\\     
&      &      &\SINR_s(\qp)\geq \gamma_s^{\mathrm{min}}, \label{P1_constriant_min_sinrs}\\
&&& ~\eqref{constraint-P^NL}-\eqref{Qos_Pr}, \label{P1_lastC}
\end{alignat}
\end{subequations}
where $\gamma_{k}^{\mathrm{min}}\triangleq 2^{\mathrm{SE}_{k}^\mathrm{min}}-1$  denotes the minimum SINR requirements
for CUs.
\begin{algorithm}[t]
\caption{Heuristic-enhanced multi-seed greedy (HE-MSG) algorithm for active STAR-RIS element selection}
\label{alg:hemsg}
\begin{algorithmic}[1]
\State \textbf{Input:}  $N$,  $N_{\mathrm{act}}$ and $\{g_n\}_{n=1}^N$.
\State \textbf{Heuristic Seed Generation:} sort $\{g_n\}$ in descending order and obtain the index set
$\mathcal{S}_{\mathrm{seed}}=\{s_1,\ldots,s_{N_{\mathrm{act}}}\}$.
\State \textbf{Initialize:} $\eta_{\mathrm{best}}^{\mathrm{overall}} \gets -\infty$, 
$\boldsymbol{\alpha}^{\mathrm{best}} \gets \mathbf{0}_N$.
\For{$m=1$ to $N_{\mathrm{act}}$} 
    \State $\mathcal{S} \gets \{s_m\}$, \quad $\mathcal{R} \gets \{1,\ldots,N\}\setminus\mathcal{S}$;
    \While{$|\mathcal{S}| < N_{\mathrm{act}}$}
        \State $\eta_{\mathrm{best}} \gets -\infty$, \quad $n^{\star} \gets 0$.
        \For{each $n \in \mathcal{R}$}
\State Construct $\boldsymbol{\alpha}$ where $\alpha_i = 1$ if $i \in \mathcal{S}\cup\{n\}$ and $\alpha_i = 0$ otherwise.
            \State Compute $\eta_{\mathrm{E}}(\boldsymbol{\alpha}) = \dfrac{\mathrm{SE}_{\mathrm{sum}}(\boldsymbol{\alpha})}{P_{\mathrm{tot}}(\boldsymbol{\alpha})}$.
            \If{$\eta_{\mathrm{E}}(\boldsymbol{\alpha}) > \eta_{\mathrm{best}}$}
                \State $\eta_{\mathrm{best}} \gets \eta_{\mathrm{E}}(\boldsymbol{\alpha})$, \quad $n^{\star} \gets n$.
            \EndIf
        \EndFor
        \State $\mathcal{S} \gets \mathcal{S} \cup \{n^{\star}\}$,
               \quad $\mathcal{R} \gets \mathcal{R} \setminus \{n^{\star}\}$.
    \EndWhile
    \State Construct $\boldsymbol{\alpha}^{(m)}$ with 
           $\alpha^{(m)}_i = 1$ if $i \in \mathcal{S}$ and $0$ otherwise.
    \State Evaluate $\eta_{\mathrm{E}}^{(m)} = \eta_{\mathrm{E}}(\boldsymbol{\alpha}^{(m)})$.
    \If{$\eta_{\mathrm{E}}^{(m)} > \eta_{\mathrm{best}}^{\mathrm{overall}}$}
        \State $\eta_{\mathrm{best}}^{\mathrm{overall}} \gets \eta_{\mathrm{E}}^{(m)}$,
               \quad $\boldsymbol{\alpha}^{\mathrm{best}} \gets \boldsymbol{\alpha}^{(m)}$.
    \EndIf
\EndFor
\State \textbf{Output:} selected vector $\boldsymbol{\alpha}^{\mathrm{best}}$.
\end{algorithmic}
\end{algorithm}
\setlength{\textfloatsep}{0.05cm}
The above problem is a classical linear fractional programming problem, which can be solved  via the Dinkelbach approach\cite{HuangAO_TWC_2024}. To this end, we can reformulate \eqref{P1} as,
\vspace{-0.6em}
\begin{subequations}\label{P1.1}
\begin{alignat}{2}
&\max_{\qp}      
&\qquad& \mathrm{SE}_{\mathrm{sum}} - \varpi_1^{(n)}\!\left(\sum\nolimits_{k \in \mathcal{K}} c_k p_k + d\right) \label{eq:P1.1:obj}\\
&\hspace{0.5em}\text{s.t.} 
&      & \eqref{P1.1b} - \eqref{P1_lastC} \label{eq:P1.1:ct1},
\end{alignat}
\end{subequations}
where $c_k\triangleq\Vert\boldsymbol{\Upsilon}_t\qG\qw_k\Vert^2+\Vert\boldsymbol{\Upsilon}_r\qG\qw_k\Vert^2$ and $d\triangleq\Vert\boldsymbol{\Upsilon}_t\Vert^2\sigma_v^2+\Vert\boldsymbol{\Upsilon}_r\Vert^2\sigma_v^2+P_t + N  (P_c + P_b)+N_\text{act}P_a$ are constant values, while $\varpi_1\ge0$ is the Dinkelbach variable, which is updated by
\vspace{-1em}
\begin{equation}~\label{eq:wnone}
  \varpi_1^{(n+1)}=\frac{ \mathrm{SE}_{\mathrm{sum}}^{(n)}}{\sum_{k \in \mathcal{K}} c_k p_k^{(n)} +d},
\end{equation}
until $\big|\mathrm{SE}_{\mathrm{sum}}^{(n)}- \varpi_1^{(n)} \big(\sum_{k \in \mathcal{K}} c_k p_k^{(n)}+d\big)\big|\leq\varepsilon$, where $\varepsilon > 0$ is a prescribed accuracy threshold for the Dinkelbach iteration, while terms with superscript $(n)$ denote the given points at iteration $(n+1)$. Accordingly, we obtain the final optimized power allocation coefficients as follows
\vspace{-0.3em}
\begin{subequations}\label{P1.2}
\begin{alignat}{2}
&\max_{\qp, v_{k_c}}      
&\qquad& \sum\nolimits_{{k_c}\in\mathcal{K}_c} v_{k_c} - \varpi_1^{(n)}\!\left(\sum\nolimits_{k \in \mathcal{K}} c_k p_k + d\right) \label{eq:P1.2:obj}\\
&\hspace{0.5em}\text{s.t.} 
&      & \mathrm{SE}_{k_c} \leq v_{k_c}, ~\forall k_c \in \mathcal{K}_c, \label{P1.1-2b}\\
&      &      & \eqref{P1.1b} - \eqref{P1_lastC} \label{eq:P1.2:ct2},
\end{alignat}
\end{subequations}
where  $v_{k_c}$ is the auxiliary variable for $\mathrm{SE}_{k_c}$. It can be observed that the optimization problem in \eqref{P1.2} is  non-convex due to constraint \eqref{P1.1-2b}. To tackle this issue,  by using the $\SINR_{k_c}$ expression in~\eqref{SINR_c}, we rewrite \eqref{P1.1-2b}  as 
\vspace{-0.4em}
\begin{align}
    f(\pckc)&\le t_{k_c}\cdot g(\qp), \quad \forall k_c\in\mathcal{K}_c,\label{productterm}
\end{align}
where  $f(\pckc)\triangleq \pck|\qh_{k_c}^T\Thetat\qG\qw^{\mathtt{c}}_{k_c}|^2$, $g(\qp)\triangleq \sum_{i \in\Kc\setminus k_c}p^{\mathtt{c}}_i|\qh_{k_c}^T\Thetat \qG \qw^{\mathtt{c}}_{i}|^2 +\sum_{k \in\Kp}p^{\mathtt{p}}_k|\qh_{k_c}^T\Thetat\qG\wpk|^2 +p_s|\qh_{k_c}^T\Thetat\qG\qw_s|^2+ \Vert\qh_{k_c}^T\mathbf{A} \circ \mathbf{\Theta}_{t}\Vert^2\sigma_v^2 + \sigma^2$, and $t_{k_c} \triangleq 2^{v_{k_c}}-1$. Since the right-hand side (RHS) of \eqref{productterm} contains product terms involving the optimization variable, the constraint remains non-convex. To address this challenge, we employ SCA and use the following upper bound 
\vspace{-0.2em}
\begin{align}
      & xy \leq 0.25 [(x+y)^2-2(x^{(n)}-y^{(n)})(x-y) 
    \nonumber\\
&\hspace{2em}+ (x^{(n)}-y^{(n)})^2],\label{xy_taylor}
\end{align}
to replace the RHS of \eqref{productterm} with its concave lower bound. 

Before we proceed, let us define 
\begin{align*}
    g_{i,k_c}&\triangleq |\qh_{k_c}^T \Thetat \qG \qw^{\mathtt{c}}_{i}|^2, \quad
    d_{k_c}\triangleq|\qh_{k_c}^T\Thetat\qG\qw^{\mathtt{c}}_{k_c}|^2,\\
    g^{\mathtt{p}}_{k,k_c} &\triangleq |\qh_{k_c}^T \Thetat \qG \wpk|^2,\quad
    g^{\mathtt{s}}_{k_c} \triangleq |\qh_{k_c}^T \Thetat \qG \qw_s|^2.
\end{align*}
Then, we can rewrite~\eqref{productterm} as
\vspace{0.4em}
\begin{align}
&4\pckc d_{k_c}\leq\!\big(t_{k_c}^{(n)}\!+\!\!p^{\mathtt{c}{(n)}}_ig_{i,k_c}\big)^2\!+\!2\bigg(t_{k_c}^{(n)}\!+\!\!\sum_{i\in\Kc\setminus k_c}p^{\mathtt{c}{(n)}}_ig_{i,k_c}\bigg) \notag\\
&\times\Big(t_{k_c}\!\!-\!t_{k_c}^{(n)}\!+\!\!\!\!\!\sum_{i\in\Kc\setminus k_c}\!\!\!\big(p^{\mathtt{c}}_i\!-p^{\mathtt{c}{(n)}}_i\big)g_{i,k_c}\Big) \!-\!\Big(t_{k_c}\!-\!\!\!\!\sum_{i\in\Kc\setminus k_c}\!\!\!p^{\mathtt{c}}_i g_{i,k_c}\Big)^2\notag\\
&+\Big(t_{k_c}^{(n)}\!\!+\!\!
\sum\nolimits_{k\in\Kp}p^{\mathtt{P}{(n)}}_k 
g^{\mathtt{p}}_{k,k_c}\Big)^2 \!+\!2\Big(t_{k_c}^{(n)}\!+\!\!\sum\nolimits_{k\in\Kp} p^{\mathtt{P}{(n)}}_k g^{\mathtt{p}}_{k,k_c}\Big) \notag\\
&\times\!\Big(t_{k_c}\!\!-\!t_{k_c}^{(n)}\!+\!\!\!\sum_{k\in\Kp}\big(p^{\mathtt{p}}_k\!-\!p^{\mathtt{P}{(n)}}_k\big) g^{\mathtt{p}}_{k,k_c}\Big) 
\!-\Big(t_{k_c}-\!\!\!\sum_{k\in\Kp}p^{\mathtt{p}}_k g^{\mathtt{p}}_{k,k_c}\Big)^2\notag\\
&+\!\big(t_{k_c}^{(n)}\!\!+\!p_s^{(n)}g^{\mathtt{s}}_{k_c}\big)^{\!2} \!\!\!+\!2\big(t_{k_c}^{(n)}\!\!+\!p_s^{(n)}g^{\mathtt{s}}_{k_c}\big)\big(t_{k_c}\!\!-\!t_{k_c}^{(n)}\!\!+\!(p_s \!\!-\! p_s^{(n)})g^{\mathtt{s}}_{k_c}\big) \notag\\
&-\big(t_{k_c}-p_s g^{\mathtt{s}}_{k_c}\big)^2+
4t_{k_c}\Big(\Vert\qh_{k_c}^T\mathbf{A} \circ \mathbf{\Theta}_{t}\Vert^2\sigma_v^2 + \sigma^2\Big),
\label{p1_xy_final}
\end{align}
\vspace{0.2em}
which is convex with respect to optimization variables. 
 
Before proceeding to solve this optimization problem, we noticed the constraint \eqref{constraint-P^NL} is equivalent to
\begin{equation}\label{psi_~}
    \Psi(Q_{k_p}) \ge \tilde{\Gamma}_E,
\end{equation}
where $\tilde{\Gamma}_E\triangleq(1-\Omega)\Gamma_E +\phi\Omega$. However, this non-convex constraint involves a sigmoidal function, which makes direct optimization highly challenging. Therefore, by applying the inverse  function transformation of the \eqref{Psi_Qk}, the constraint can be reformulated into a more tractable form:
\begin{equation}\label{Q_psi}
    Q_{k_p} = \chi - \frac{1}{\xi} \ln\left(\frac{\phi - \Psi}{\Psi}\right).
\end{equation}
By applying \eqref{psi_~} and \eqref{Q_psi}, constraint \eqref{constraint-P^NL} is reformulated as a lower bound constraint as
\begin{equation}\label{Qk_psi2}
   Q_{k_p} \ge  Q_{k_p}^{\mathrm{min}}\triangleq\chi - \frac{1}{\xi} \ln\bigg(\frac{\phi - \tilde{\Gamma}_E}{\tilde{\Gamma}_E}\bigg) .
\end{equation}

From the above discussion, we arrive at the following the approximate convex problem
\begin{subequations}\label{P1.1-4}
\begin{alignat}{2}
&\max_{\qp,\{t_{k_c}\}}      
&~\hspace{0.1em}& \sum\nolimits_{k_c}\!\!\!\log_2(1+t_{k_c}) \!-\! \varpi_1^{(n)}\!\Big(\sum\nolimits_{k \in \mathcal{K}}\!\! c_k p_k \!+\! d\Big) \label{eq:P1.1-4:obj}\\
&\hspace{0.2em}\text{s.t.} 
&      & \log_2(1+t_{k_c}) \ge \mathrm{SE}_{k_c}^\mathrm{min},\\
& & & \eqref{Qos_Pr},\eqref{P1.1b},  \eqref{p1_xy_final},\eqref{P1constriant_min_sinrkc}, \eqref{P1_constriant_min_sinrs}, \eqref{Qk_psi2}. \label{eq:P1.1-4:ct1}
\end{alignat}
\end{subequations}
To this end, problem \eqref{P1.1-4} is a convex optimization problem, which can be efficiently tackled with standard convex solvers, such as CVX. The algorithm details are shown in \textbf{Algorithm~\ref{alg:P1}}.
\begin{algorithm}[t]
\caption{Proposed iterative algorithm for solving power allocation subproblem}
\label{alg:P1}
\begin{algorithmic}[1]
\State \textbf{Initialize:} feasible power allocation $\qp^{(0)}$, auxiliary variables 
$\{t_{k_c}^{(0)}\}$, Dinkelbach parameter $\varpi_1^{(0)}=0$, iteration index $n=0$, 
and error tolerance $\varepsilon$.
\Repeat
    
    \State Solve the convex problem \eqref{P1.1-4} to obtain $\qp^{(n+1)}$ and $\{t_{k_c}^{(n+1)}\}$.
    \State Update $\varpi_1^{(n+1)}$ as in~\eqref{eq:wnone}.
    \State $n =n+1$.
\Until{convergence  $|\mathrm{SE}_{\mathrm{sum}}^{(n)}- \varpi_1^{(n-1)}\sum_{k \in \mathcal{K}} c_k p_k^{(n)} +d)|\leq \varepsilon$ or $n>$ max iteration numbers.}
\State \textbf{Output:} optimal power allocation $\qp^\star$.
\end{algorithmic}
\end{algorithm}
\setlength{\textfloatsep}{0.25cm}

\vspace{-1.4em}
\subsection{STAR-RIS Coefficient Design}
In this subproblem, we focus on EE optimization by adjusting the STAR-RIS coefficients under the given power allocation. Mathematically, this can be formulated as follows:
\vspace{-0.6em}
\begin{subequations}\label{P2}
\begin{alignat}{2}
&\ \max_{\mathbf{\Theta}_t,\mathbf{\Theta}_r}      
&\qquad& \eta_{\mathrm{E}} = \frac{\mathrm{SE}_{\mathrm{sum}}(\mathbf{\Theta}_t)}{P_{\mathrm{tot}}} \\ 
&\hspace{1em}\text{s.t.} 
&      & \mathrm{SINR}_{k_c}(\mathbf{\Theta}_t) \ge \gamma_{k_c}^\mathrm{min}, \label{P1.2.SINRkc}
\\
&      &      & \SINR_s(\mathbf{\Theta}_r) \ge \gamma_s^\mathrm{min},\label{p1.2.SINRs}
\\
&      &      &
 Q_{k_p}(\mathbf{\Theta}_t)  \ge  Q_{k_p}^{\mathrm{min}},\label{p1.2.Qptetat}
 \\
&&&~\eqref{Qos_Pr},\eqref{RIS_law1}-\eqref{RIS_law2}.
\end{alignat}
\end{subequations}
This optimization problem is non-convex since both the objective function and the constraints~\eqref{P1.2.SINRkc}, \eqref{p1.2.SINRs}, \eqref{p1.2.Qptetat}, and \eqref{Qos_Pr} are non-convex with respect to the optimization variables. To deal with this, we define the new variables as $\qv_{q} = [v_{q,1}, v_{q,2}, \dots, v_{q,N}]^\mathrm{T} \in \mathbb{C}^{N \times 1}$, with $q\in\{t,r\}$ 
where $\operatorname{diag}(\qv_t) \triangleq \boldsymbol{\Theta}_t$ and $\operatorname{diag}(\qv_r) \triangleq \boldsymbol{\Theta}_r$. 

We next reformulate constraint \eqref{P1.2.SINRkc} as
\begin{align} \label{gammaI_c}
 \gamma_{k_c}^{\mathrm{min}}  I_{k_c}(\qv_t)  \le \pckc \left|  \qv_t^T\qh_{k_c} \qG \qw^{\mathtt{c}}_{k_c} \right|^2,
\end{align}
where $\qh_{k_c}\triangleq\operatorname{diag}(\qh_{k_c} )$ and
\begin{align}
I_{k_c}&(\qv_t) \triangleq \!  \! \! \! \! 
\sum_{i\in\mathcal{K}_c\setminus k_c} p^{\mathtt{c}}_i \left| \qv_t^T\qh_{k_c}\qG \qw^{\mathtt{c}}_i \right|^2
\! + \! \! \! \sum_{k \in\mathcal{K}_p} p^{\mathtt{p}}_k \left|\qv_t^T \qh_{k_c} \qG \wpk \right|^2 \notag \\
& + p_s \left|\qv_t^T \qh_{k_c} \qG \qw_s \right|^2
+ \left\Vert \qh^T_{k_c} \mathbf{A}\operatorname{diag}(\qv_t) \right\Vert^2 \sigma_v^2 + \sigma^2.
\label{I_vt_zt}
\end{align}

To address the non-convexity of constraint \eqref{P1.2.SINRkc}, we employ the SCA method to derive the following concave lower bound:
\begin{align} 
&\gamma_{k_c}^{\mathrm{min}}I_{k_c}(\qv_t) 
\leq \pckc \big|  \vtmT\qh_{k_c} \qG \qw^{\mathtt{c}}_{k_c} \big|^2 +2\pckc\Re\big\{(\vtmT\qh_{k_c}\notag\\
&\hspace{5em} \times\qG \qw^{\mathtt{c}}_{k_c}(\qw^{\mathtt{c}}_{k_c})^H\qG^H\qh_{k_c}^H)\big(\qv_t-\qv_t^{(m)}\big)\big\},
\label{Gamma_kc_final}
\end{align}
\vspace{0.2em}
where we used a superscript $(m)$ to denote
the value of the involving variable produced after $(m-1)$ iterations
($m\geq1$).

The non-convex constraint \eqref{p1.2.SINRs} can be rewritten as
\begin{equation}\label{gammaI_s}
    \gamma_s^{\mathrm{min}}  I_s(\qv_r) \le p_s \Vert\beta_s\qb_s\qa^T_\text{RT} \operatorname{diag}(\qv_r)  \qG \qw_s \Vert^2,
\end{equation}
where
\begin{align}
I_s(\qv_r) \triangleq &
\sum\nolimits_{i \in \mathcal{K}_c} p^{\mathtt{c}}_i \Vert \beta_s\qb_s\qv_r^T\operatorname{diag}(\qa_\text{RT})  \qG \qw^{\mathtt{c}}_i \Vert^2\notag\\
&+ \sum\nolimits_{k \in \mathcal{K}_p}p^{\mathtt{p}}_k \Vert \beta_s\qb_s\qv_r^T\operatorname{diag}(\qa_\text{RT})  \qG \wpk \Vert^2 \notag\\
&+ \Vert \beta_s\qb_s\qa^T_\text{RT} \qA \qv_r \Vert^2 \sigma_v^2 + \sigma^2.
\end{align}
\vspace{0.1em}

Nevertheless, constraint \eqref{gammaI_s} remains non-convex. Therefore, we leverage the SCA technique to derive the following concave lower bound:
\begin{align}
 \gamma_s^{\mathrm{min}}I_s(\qv_r) &\le    p_s \big| \beta_s \qb_s \vrmT    \operatorname{diag}(\qa_{\text{RT}}) \qG \qw_s \big|^2 \notag\\
 &\quad+2p_s\Re\big\{\big(\beta_s^2 \qb_s^H\qb_s \vrmT    \operatorname{diag}(\qa_{\text{RT}}) \qG \qw_s\notag \\
&\quad\times\qw_s^H\qG^H\operatorname{diag}(\qa_{\text{RT}}^H) \big)(\qv_r-\qv_r^{(m)})\big\}.
 \label{gamma_s_final}
\end{align}
\vspace{0.1em}

Next, we focus on constraint \eqref{p1.2.Qptetat}, whose non-convexity arises from the non-convex nature of ${Q}_{k_p}$ with respect to the optimization variables. To address this, we first express ${Q}_{k_p}$ in \eqref{Q_kp} in terms of $\qv_t$ as follows:
\begin{align}
 {Q}_{k_p} &=   \sum\nolimits_{i\in\Kc} \pci \left| \qv_t^T\operatorname{diag(\qh_{k_p})} \qG \qw^{\mathtt{c}}_i \right|^2
\nonumber\\
&\hspace{-2em}
+ \sum\nolimits_{k\in\Kp} p^{\mathtt{p}}_k\left| \qv_t^T\operatorname{diag(\qh_{k_p})}\qG\wpk \right|^2\nonumber\\
&\hspace{-2em}
+ p_s\left| \qv_t^T\operatorname{diag(\qh_{k_p})}\qG\qw_s \right|^2 + \Vert \qh_{k_p}^T \mathbf{A}  \operatorname{diag(\qv_t)} \Vert^2 \sigma_v^2. 
\end{align}
\vspace{0.2em}
This expression remains non-convex. Using the SCA, we can derive a convex lower bound as shown in \eqref{Qkp_taylor} at the top of next page.
\begin{figure*}
\begin{align}\label{Qkp_taylor}
 {Q}_{k_p}&\!\geq \!\tilde{Q}_{k_p}\!\triangleq\sum\nolimits_{i\in\Kc} \!\!\pci \Big| \vtmT\operatorname{diag(\qh_{k_p})} \qG \qw^{\mathtt{c}}_i \Big|^2 \!+\!\sum\nolimits_{i\in\Kc}\!\! 2\pci\Re\Big\{\vtmT\operatorname{diag(\qh_{k_p})} \qG \qw^{\mathtt{c}}_i(\qw^{\mathtt{c}}_i)^H\qG^H\operatorname{diag(\qh_{k_p})}^H\big(\qv_t\!-\!\qv_t^{(m)}\big)\Big\} \notag \\
&\quad+
\sum\nolimits_{k\in\Kp} p^{\mathtt{p}}_k
\Big| \vtmT\operatorname{diag}(\qh_{k_p})\qG\wpk \Big|^2
+ 
\sum\nolimits_{k\in\Kp} 2p^{\mathtt{p}}_k
\Re\Big\{
\vtmT\operatorname{diag}(\qh_{k_p})\qG\wpk
(\wpk)^H\qG^H\operatorname{diag}(\qh_{k_p})^H
\big(\qv_t - \qv_t^{(m)}\big)
\Big\}\notag \\
&\quad +p_s\Big| \vtmT \operatorname{diag}(\qh_{k_p}) \qG \qw_s \Big|^2 +2p_s  \Re \Big\{
\vtmT \operatorname{diag}(\qh_{k_p}) \qG \qw_s
\cdot \qw_s^H \qG^H \operatorname{diag}(\qh_{k_p})^H \big(\qv_t - \qv_t^{(m)}\big)
\Big\}\notag \\
&\quad+\sigma_v^2 \big\| \qh_{k_p}^T \mathbf{A} \operatorname{diag}(\qv_t^{(m)}) \big\|^2 
+ 2\sigma_v^2  \Re \Big\{
\left( \qh_{k_p}^T \mathbf{A} \operatorname{diag}(\qv_t^{(m)}) \right) 
\left( \operatorname{diag}\left( \mathbf{A}^H \qh_{k_p}^* \right) \right)
 \big(\qv_t - \qv_t^{(m)}\big)
\Big\}.
\end{align}
\hrulefill
\vspace{-1.5em}
\end{figure*} 
Thus,~\eqref{Qk_psi2} can be reformulated as the following convex constraint
\vspace{-0.5em}
\begin{equation}
     \tilde{Q}_{k_p} (\qv_t)\ge \chi - \frac{1}{\xi} \ln\left(\frac{\phi - \tilde{\Gamma}_E}{\tilde{\Gamma}_E}\right). \label{Q_min_p2}
\end{equation}

Moreover, the convex reformulation of constraint \eqref{Qos_Pr} is 
\vspace{-0.3em}
\begin{align}
    &\sum\nolimits_{k\in \mathcal{K}}p_k \left\|  \qA\operatorname{diag}(\qv_t)\qG \right\|^2
+ \sum\nolimits_{k \in \mathcal{K}}p_k \left\| \qA \operatorname{diag}(\qv_r)\qG \right\|^2\notag\\
&+ \| \qA\operatorname{diag}(\qv_r)  \|^2 \sigma_v^2
+ \| \qA\operatorname{diag}(\qv_t)  \|^2 \sigma_v^2
\leq P_r^{\mathrm{max}}. \label{Pr_final} 
\end{align}

Finally, to handle the non-convexity of the objective function, we introduce an auxiliary variable $\vartheta_j$ as follows
\vspace{-0.2em}
\begin{subequations}\label{5vars}
\begin{alignat}{2}
&\max_{\qv_t,\qv_r,\{\vartheta_j\}}      
&\quad&  \eta_{\mathrm{E}}=
\frac{\sum_{j\in\mathcal{K}_c}\vartheta_j}{P_{\text{tot}}}
            \\
&\hspace{1em}\text{s.t.} 
&      &\vartheta_j \ge \operatorname{log}_2(1+\SINR_{j}(\qv_t)), \forall j\in \mathcal{K}_c,\label{R_kc.P.2}\\
&      &      &  0<|[\qv_q]_n|\leq \rho_\mathrm{max},~q\in\{t,r\}, \forall n\in\mathcal{A},\label{vt_r01}\\
&      &      &  0<|[\qv_q]_n|\leq 1,~q\in\{t,r\},\forall 
\ n\notin\mathcal{A},\\
&      &      &  |[\qv_r]_n|^2 + |[\qv_t]_n|^2\leq 1, 
\ n\notin\mathcal{A},\\
&      &      & \eqref{Gamma_kc_final},\eqref{gamma_s_final},\eqref{Pr_final},\eqref{Q_min_p2}.\label{60final}
\end{alignat}
\end{subequations}
Problem~\eqref{5vars} remains an intractable problem due to the non-convex terms in~\eqref{R_kc.P.2}. To tackle this issue, we introduce the following auxiliary variables $x_j$ and $y_j$ ($\forall j\in\mathcal{K}_c$) as
\vspace{-0.3em}
\begin{align}
&\frac{1}{x_j}\le  \pcj \big|  \vtmT\qh_{j} \qG \wcj\big|^2\notag\\
&+2\pcj\Re\Big\{(\vtmT\qh_{j} \qG \wcj\wcjH\qG^H\qh_{j}^H)(\qv_t-\qv_t^{(m)})\Big\},\label{xj}\\
y_j& \!\ge\!\! 
\sum\nolimits_{i \in\mathcal{K}_c\setminus j} p^{\mathtt{c}}_i \left| \qv_t^T\qh_{j}\qG \qw^{\mathtt{c}}_i \right|^2
\!+\! \sum\nolimits_{k \in \mathcal{K}_p} p^{\mathtt{p}}_k \left|\qv_t^T \qh_{j} \qG \wpk \right|^2 \notag 
\\
& + p_s \left|\qv_t^T \qh_{j} \qG \qw_s \right|^2
+ \left\Vert \qh_{j}^T \mathbf{A}\operatorname{diag}(\qv_t) \right\Vert^2 \sigma_v^2 + \sigma^2.\label{yj}
\end{align}
 Following \cite{Mohammadi:FD-CF-MIMO:JSAC:2023}, for any given points $x_j^{(m)}$ and $  y_j^{(m)}$, the constraint \eqref{R_kc.P.2} can attain the lower bound as
 \vspace{-0.4em}
    \begin{align}
       \vartheta_j \ge& \operatorname{log}_2\Big(1+\frac{1}{x_j^{(m)}y_j^{(m)}}\Big)-\frac{x_j-x_j^{(m)}}{\ln(2)\big(x_j^{(m)}+x_j^{(t)^2}y_j^{(m)}\big)}\notag\\
        &-\frac{y_j-y_j^{(m)}}{\ln(2)\Big(y_j^{(m)}+y_j^{(t)^2}x_j^{(m)}\Big)} \triangleq \bar\vartheta_j, ~j\in\mathcal{K}_c.\label{R_kc_final}
    \end{align}

To this end, problem~\eqref{P2} is transformed
into 
\vspace{-0.2em}
\begin{subequations}\label{P2_final}
\begin{alignat}{2}
&\max_{\qv_t,\qv_r,\{\vartheta_j\}}      
&\qquad& \frac{\sum_{j\in\mathcal{K}_c} \vartheta_j}{P_{\text{tot}}} 
            \\
&\hspace{1.5em}\text{s.t.} 
 &      & \eqref{vt_r01}-\eqref{60final},\eqref{xj}-\eqref{R_kc_final}.
\end{alignat}
\end{subequations}

The optimization problem~\eqref{P2_final} is  a  convex optimization problem and can be solved using standard convex solvers, e.g., CVX\cite{Mumtaz:EH-CF-MIMO:TGCN:2025}. 
\begin{algorithm}[t]
\caption{STAR-RIS coefficient design }
\label{alg:P2}
\begin{algorithmic}[1]
\State \textbf{Initialize:} feasible $(\qv_t^{(0)},\qv_r^{(0)},\{R_j^{(0)}\})$, 
set $\varpi_2^{(0)}=0$, iteration index $m=0$, tolerance $\varepsilon$.

  \State For given $\qp$ and solve optimization problem \eqref{P2_final}.
  \State \textbf{Output:} Solve convex problem to obtain updated $\qv_t^{(m+1)},\qv_r^{(m+1)},x_{j}^{(m+1)},y_{j}^{(m+1)}$,$\eta_{\mathrm{E}}^{(m+1)}$.\\The physically realizable STAR-RIS coefficients $\boldsymbol{\Theta}_t^{(m+1)}$ and $\boldsymbol{\Theta}_r^{(m+1)}$ are constructed from $\qv_t^{(m+1)}$ and $\qv_r^{(m+1)}$ following the predefined mapping.
  
\end{algorithmic}
\end{algorithm}
\setlength{\textfloatsep}{0.05cm}

\vspace{-0.8em}
\subsection{Overall Algorithm and Complexity Analysis}

Finally, we employ an outer AO procedure, summarized in \textbf{Algorithm~\ref{alg:AO}}, to jointly optimize all design parameters. Specifically, the procedure first executes \textbf{Algorithm~\ref{alg:P1}} to update the power-allocation vector, and then applies \textbf{Algorithm~\ref{alg:P2}} using the updated variables to refine the phase-shift coefficients. This process is iteratively repeated until the improvement in the objective function between consecutive iterations falls below a predefined threshold $\varepsilon_\text{AO}$.

\begin{algorithm}[t]
\caption{Overall AO framework}
\label{alg:AO}
\begin{algorithmic}[1]
\State \textbf{Initialize:} feasible $(\Thetat^{(0)},\Thetar^{(0)},\qp^{(0)})$, set $t=0$, tolerances $\varepsilon_\text{AO}$, max iteration numbers $T_\text{AO}$.
\State Obtain the selection vector $\boldsymbol{\alpha}$ using Algorithm~3.
\Repeat
  \State \textbf{(Power Allocation)}: Fix $(\Thetat^{(t)},\Thetar^{(t)})$, solved $\Rightarrow \qp^{(t+1)}$.
  \State \textbf{ (Mode \& Coefficients)}: Fix $\qp^{(t+1)}$, solved $\Rightarrow (\Thetat^{(t+1)},\Thetar^{(t+1)})$.
  \State Update $\eta_{\mathrm{E}}^{(t+1)}$.
  \State $t \gets t+1$.
\Until{$|\eta_{\mathrm{E}}^{(t)}-\eta_{\mathrm{E}}^{(t-1)}|\le \varepsilon_\text{AO}$ or $t\ge T_\text{AO}$}
\State \textbf{Output:} $(\Thetat^\star,\Thetar^\star,\qp^\star,\eta_{\mathrm{E}}^\star)$.
\end{algorithmic}
\end{algorithm}
\setlength{\textfloatsep}{0.05cm}

\begin{table}[t]
\setlength{\abovecaptionskip}{0cm}
\captionsetup{justification=centering}
\centering
\caption{Computational complexity variables for optimization problems }
\centering
\small
\label{table:comlexity}
\begin{tabular}{|l|l|l|l|}
\hline
Problem & $A_v$        & $A_l$     & $A_q$       \\ \hline
\textbf{Algorithm 4}      & $2K_C+K_P+1$ & $K_C+K_P$ & $K_C+K_P+1$ \\ \hline
\textbf{Algorithm 5}      & $K_C$      & $2N+K_C$  & $K_C+K_P+1$ \\ \hline
\end{tabular}
\end{table}

\begin{table}[t]

\setlength{\abovecaptionskip}{0cm}
\caption{Simulation parameters}\label{parameters}
\centering
\scriptsize
\begin{tabular}{|p{0.9cm}|p{4.8cm}|p{1.1cm}|}
\hline
\textbf{Notation} & \centering\textbf{Parameter} & \textbf{Value} \\ \hline
$P_c + P_b$   & Static circuit and bias power per RIS element & 10 dBm \\ \hline
$P_a$         & Power consumed by each active element         & $0.1$ W  \\ \hline
$P_t$         & Transmit power of the BS                      & 40 dBm \\ \hline
$R_{\min}$    & Minimum SE requirement of each CU & 0.1 bps/Hz \\ \hline
$Q_{\min}$    & Minimum harvested power for each EU & $1\,\mu\text{W}$  \\ \hline
$\gamma_{s}^{\min}$ & Minimum sensing SINR requirement & 2 dB \\ \hline
$\Lambda$ & Radar cross section & $1\,\text{m}^2$ \\ \hline
$T_s$ & Chirp duration & $5\,\mu\text{s}$ \\ \hline
$\kappa$      & Ricean factor of BS–RIS and RIS–user channels & 3 dB \\ \hline
$\sigma_v^2$  & Thermal noise generated by active RIS elements & -40 dBm \\ \hline
$\sigma^2, \sigma_s^2$ & Noise power at receivers and sensors & -40 dBm \\ \hline
$\varepsilon$, $\varepsilon_\text{AO}$ & Convergence accuracy & $10^{-3}$ \\ \hline

\end{tabular}

\vspace{-0.1em}
\normalsize
\end{table} 

We relaxed  optimization \textbf{Algorithm~\ref{alg:P1}} and \textbf{Algorithm~\ref{alg:P2}} requires a computational complexity of $\mathcal{O}\!\left( \sqrt{A_l + A_q} \left( A_v + A_l + A_q \right) A_v^{2} \right)$ per iteration\cite{Mohammadi:FD-CF-MIMO:JSAC:2023},\cite{Tam:Load-Balancing:TWC:2017}, where  $A_v$ denotes the total number of real-valued scalar optimization variables, while $A_l$ and $A_q$ correspond to the numbers of linear and quadratic constraints, respectively. These variables are shown in \textbf{Table~\ref{table:comlexity}} for all considered optimization problems.

\vspace{-0.5em}
\section{Numerical Results} \label{sec:Numerical}
In this section, the numerical results obtained through 
Monte-Carlo simulations are provided to evaluate the performance of the proposed hybrid STAR-RIS-enabled system and the associated optimization scheme.

The hybrid STAR-RIS is positioned at $(0,0,8)\,\mathrm{m}$, while the BS is located at $(-30,-30,15)\,\mathrm{m}$. The target is placed at $(5,-5,\,5)\,\mathrm{m}$. The CUs are randomly distributed within a distance range of $20$--$50\mathrm{m}$ from the transmission side of the STAR-RIS. Similarly, the EUs are randomly distributed within $10$--$20\mathrm{m}$ from the transmission side of the STAR-RISs. For the number of RIS elements, we set $N_h=8$ and increase $N_v$ linearly with $N$. Similarly, for the number of  sensors, we set $N_{s,x}=3$ and increase $N_{s,y}$ linearly with $N_s$. The carrier frequency is $f_c=2.4$ GHz. Without loss of generality, we assume that all CUs have the same QoS requirement, i.e., $\gamma_{k_c}^{\mathrm{min}}=\gamma^\mathrm{min}$. In the nonlinear EH model, the maximum harvested power is set to $\phi = 4\,\mu\text{W}$, with shaping parameters $\xi = 1.5 \times 10^6$ and $\chi = 2.2 \times 10^{-6}$ governing the transition. Other simulation parameters appear in \textbf{Table~\ref{parameters}}.

\begin{figure}[t]
    \centering
    \includegraphics[width=0.43\textwidth]{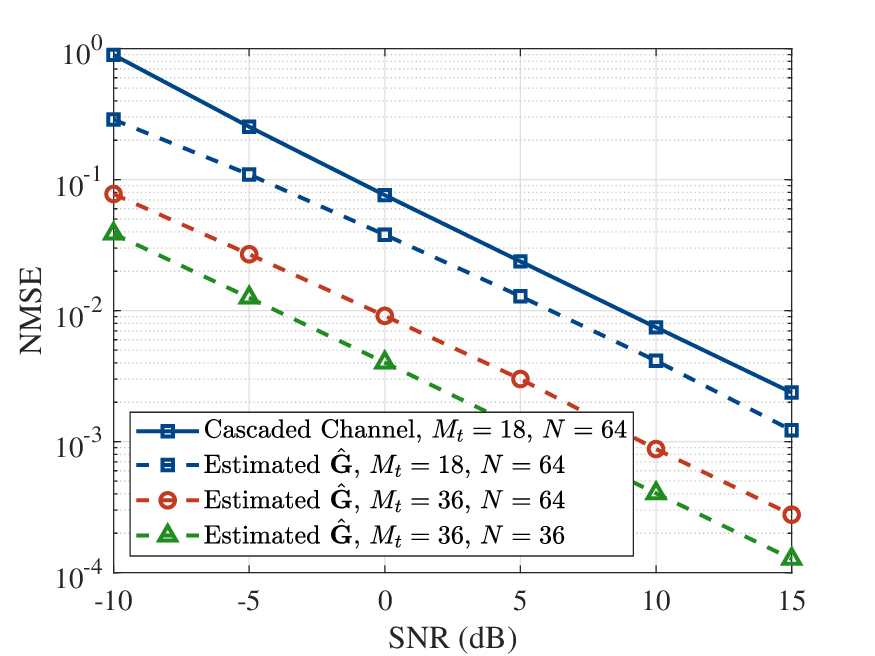}
    \caption{ NMSE performance of the proposed channel estimation method versus the SNR.}
    \label{fig:placeholder}
\end{figure}

Figure~\ref{fig:placeholder} evaluates the normalized mean-square error (NMSE) performance of the proposed channel estimation scheme under different system configurations. For the case of $M_t=18$ and $N=64$, the estimated factor matrix $\hat{\mathbf{G}}$ is evaluated after applying a column-wise scaling alignment to compensate for the inherent scaling ambiguity of the PARAFAC decomposition. As observed from the figure, the NMSE of the cascaded channel is higher than that of $\hat{\mathbf{G}}$. This is because the cascaded channel is reconstructed from the product of two estimated channel factors, such that the scaling ambiguities and residual estimation errors associated with both factors are propagated and accumulated during the reconstruction process. Furthermore, for a fixed number of STAR-RIS elements ($N=64$), increasing the number of BS antennas improves the channel estimation accuracy. This improvement can be attributed to the higher tensor rank and the additional observation diversity provided by a larger number of BS antennas, which facilitate more reliable tensor factorization. Specifically, increasing $M_t$ from 18 to 36 reduces the NMSE by approximately $76.7\%$. In addition, for a fixed number of BS antennas ($M_t=36$), reducing the number of STAR-RIS elements from $N=64$ to $N=36$ leads to a lower-dimensional channel estimation problem. As a result, the tensor recovery process becomes less challenging, yielding an NMSE improvement of approximately $57.6\%$. These results demonstrate that the proposed channel estimation scheme can achieve reliable estimation performance across different BS antenna and STAR-RIS configurations.

Figure~\ref{fig:MUSIC_AoA} illustrates the MUSIC spectrum under different sensor-array sizes and sensing SINR levels. With a small array $(N_s = 3\times3)$, the spectrum exhibits a broad mainlobe and noticeable sidelobes, leading to a clear AoA estimation bias. This is because the limited aperture results in poor angular resolution and stronger angle coupling in the two angles dimensional  domains. Increasing the array size to $5\times5$ sharpens the peak and suppresses sidelobes, thereby yielding an estimate that is nearly aligned with the ground truth. Moreover, a higher sensing SINR improves the separation between signal and noise subspaces, producing a cleaner spectrum and substantially reducing estimation error even with a smaller array. These results confirm that sensing aperture and SINR critically affect MUSIC resolution, with larger arrays or higher SINR consistently improving AoA accuracy.

\begin{figure*}[t]
    \centering
    \begin{subfigure}[b]{0.32\textwidth}
        \includegraphics[width=\textwidth,trim=0 0 0 10, clip]{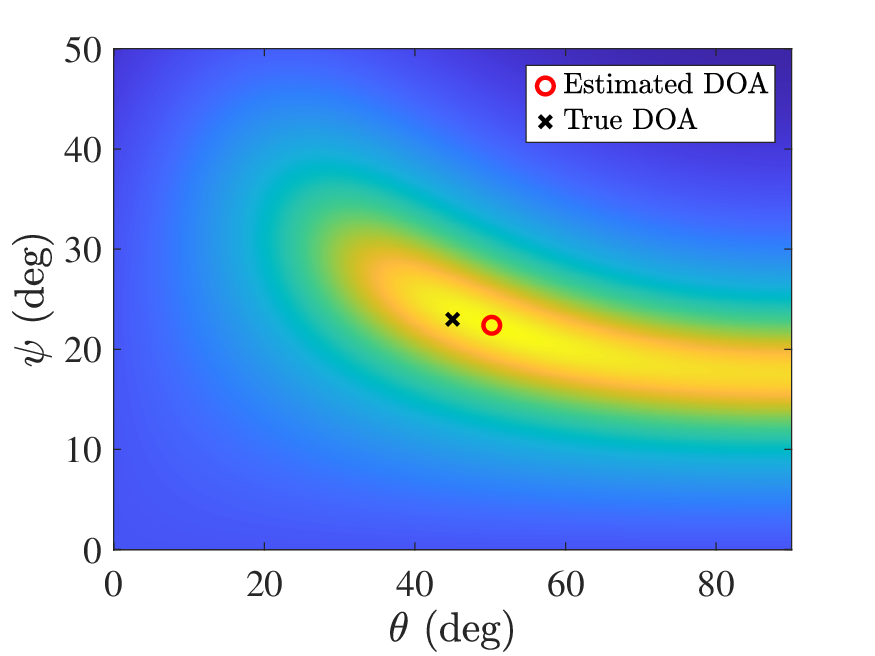}
        \caption{\small $N_s = 3 \times 3$, $\gamma_s = 8$ dB}
        \label{fig:MUSIC_8_33}
    \end{subfigure}
    \hfill
    \begin{subfigure}[b]{0.32\textwidth}
        \includegraphics[width=\textwidth,trim=0 0 0 10, clip]{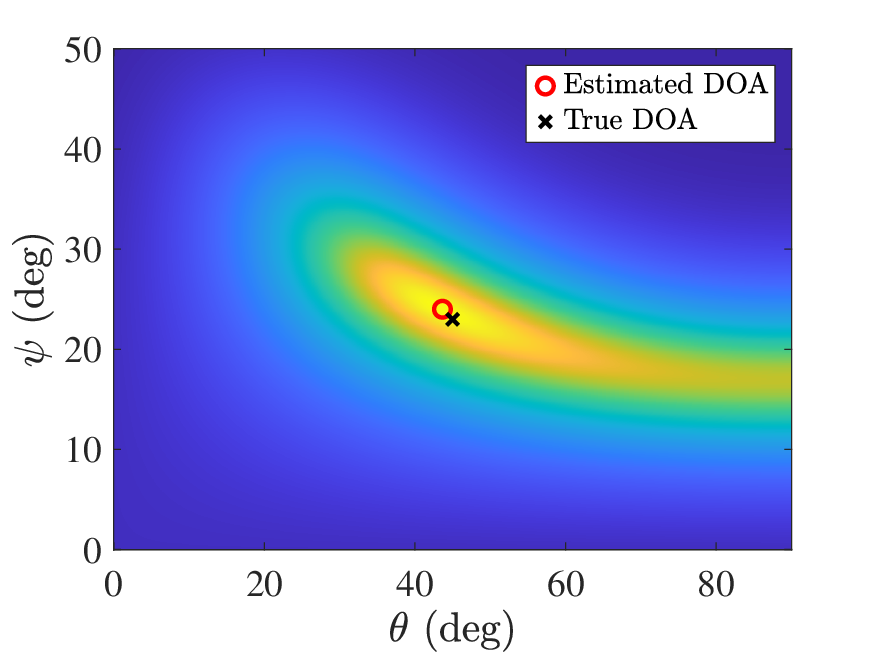}
        \caption{\small $N_s = 5 \times 5$, $\gamma_s = 8$ dB}
        \label{fig:MUSIC_8_55}
    \end{subfigure}
    \hfill
    \begin{subfigure}[b]{0.32\textwidth}
        \includegraphics[width=\textwidth,trim=0 0 0 10, clip]{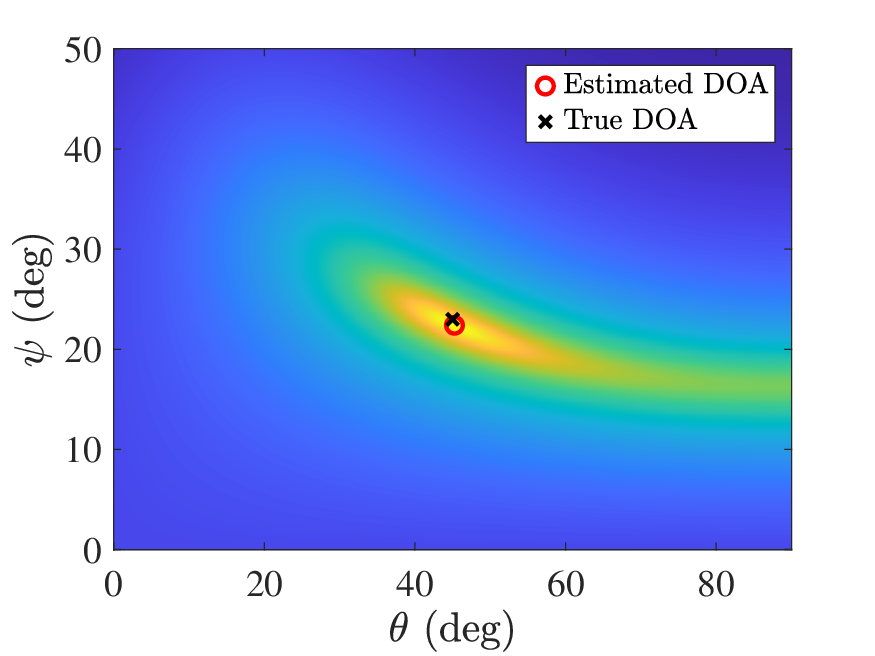}
        \caption{\small $N_s = 3 \times 3$, $\gamma_s = 15$ dB}
        \label{fig:MUSIC_15_33}
    \end{subfigure}
    \caption{\small MUSIC spectrum for AoA estimation under varying array sizes and sensing SINRs ($M=48$, $N=64$).}
    \label{fig:MUSIC_AoA}
    \vspace{-0.1em}
\end{figure*}

\begin{figure*}[t]
\vspace{-0.2cm}
    \centering
    \begin{minipage}[t]{0.32\textwidth}
        \centering
        \includegraphics[width=\textwidth,trim=0 -10 0 17,clip]{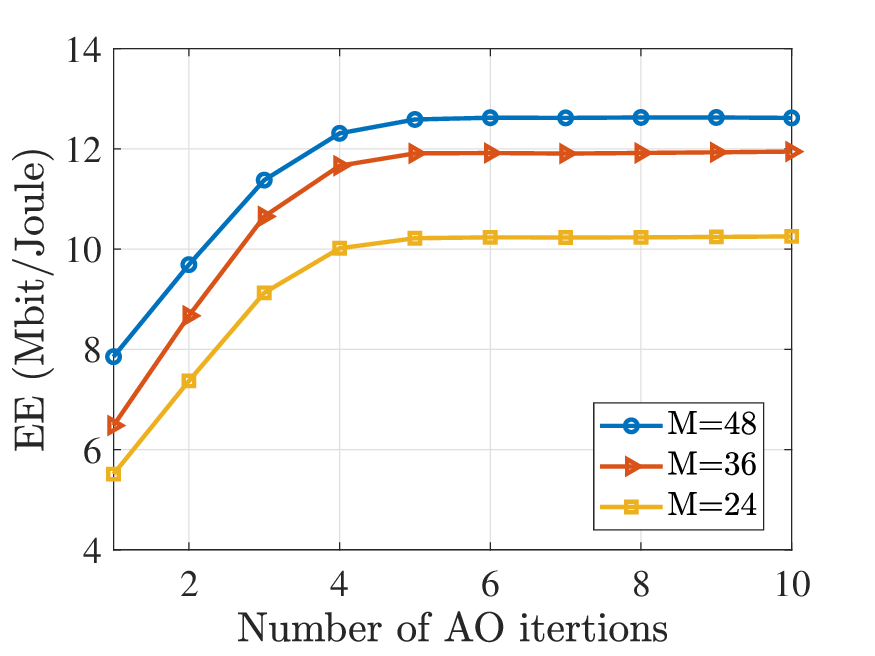}
        \vspace{-1.9em}
        \caption{\small  Convergence of \textbf{Algorithm~\ref{alg:AO}} ($N=64, K_p=2, K_c=2$).}
        \label{fig:AO_iter}
    \end{minipage}
    \hfill
    \begin{minipage}[t]{0.32\textwidth}
        \centering
        \includegraphics[width=\textwidth,trim=0 -10 0 17,clip]{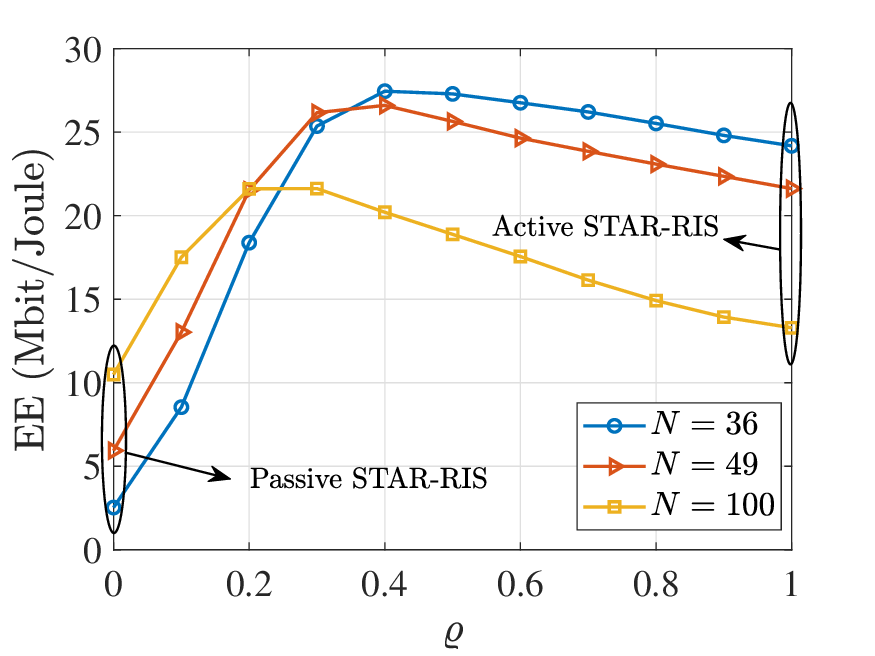}
        \vspace{-1.9em}
        \caption{\small  EE versus the active elements' ratio ($M=100, K_p=2, K_c=2$).}
        \label{fig:EE_ratio}
    \end{minipage}
    \hfill
    \begin{minipage}[t]{0.32\textwidth}
        \centering
        \includegraphics[width=\textwidth,trim=0 -10 0 17,clip]{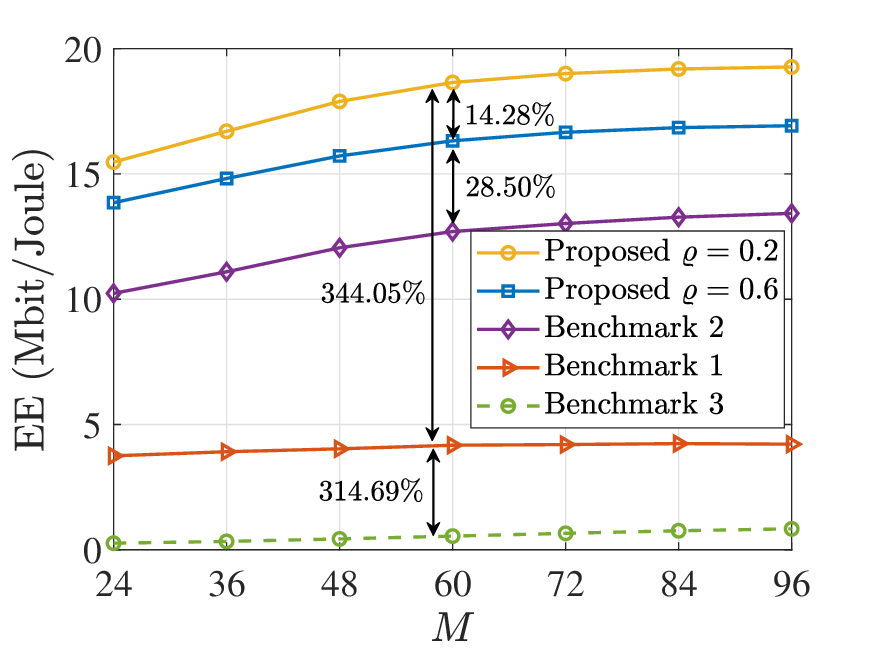}
        \vspace{-1.9em}
        \caption{\small  Performance under different activation strategies ($N=64, K_p=2, K_c=2$).}
        \label{fig:EEM}
    \end{minipage}
\vspace{-1.6em}
\end{figure*}
Figure~\ref{fig:AO_iter} presents the convergence performance of the proposed  \textbf{Algorithm~\ref{alg:AO}} under different BS antenna configurations with $M=24$, $36$, $48$. The algorithm demonstrates fast and stable convergence, reaching near-optimal EE within approximately $4$–$5$ iterations for all cases, confirming the efficiency of the Dinkelbach–SCA hybrid structure.
Increasing the number of BS antennas significantly improves the EE since a larger array provides stronger beamforming gain and more spatial degrees of freedom for interference mitigation and RIS-assisted cascaded channel shaping.
Furthermore, the EE increases monotonically throughout the AO iterations without oscillations, highlighting the robustness of the proposed convex approximations and the block coordinate ascent structure. These results collectively confirm that \textbf{Algorithm~\ref{alg:AO}} is both computationally efficient and stable.

Figure~\ref{fig:EE_ratio} shows how the active-element ratio influences the system EE, leading to two key observations: i) Increasing the number of active elements initially boosts the EE because moderate activation significantly enhances STAR-RIS controllability and improves SINR. It is worth noting that $\varrho=1$ corresponds to the fully active STAR-RIS architecture, in which all STAR-RIS elements operate in the active mode. Across all considered $N$ values, the EE rises sharply as $\varrho$ increases from $0$ to about $0.2$. For example, when $N = 100$, the EE grows from below $5$~Mbit/Joule to over $20$~Mbit/Joule. Beyond this point, the SINR gains saturate, and the additional hardware power consumed by active components—such as amplifiers, bias circuits, and related circuitry—begins to dominate the total power consumption. Consequently, further activation beyond $\varrho = 0.3$ yields only marginal EE gains, while hardware power continues to escalate. Ultimately, the  EE declines from its peak ($22$--$26$~Mbit/Joule) to below $15$~Mbit/Joule as $\varrho$ approaches $1$, highlighting that hardware overhead imposes a fundamental limit on efficiency.
ii) The optimal active ratio varies with the RIS size. As $N$ increases, the EE peak occurs at a smaller active ratio because the SINR gain from additional active elements is limited, while their circuit power consumption grows. For instance, the $N = 49$ configuration achieves its maximum EE at a slightly higher ratio (around $\varrho = 0.3$), as a smaller RIS surface requires more active elements. In contrast, the larger $N = 100$ surface reaches peak EE with only $\varrho = 0.2$. The gap between the EE peak and the EE at $\varrho = 1$ widens with increasing $N$, indicating that large STAR-RISs should use a moderate rather than fully active ratio.

Figure~\ref{fig:EEM} shows the EE performance versus the number of BS antennas under different activation ratios for various designs. Specifically, the following designs are considered:
\begin{itemize}
    \item \textbf{Proposed:} Hybrid STAR-RIS configured using \textbf{Algorithm~\ref{alg:AO}} under varying activation ratios.
    \item \textbf{Benchmark~1:} Hybrid STAR-RIS with random phase-shift design and  power allocation.
    \item \textbf{Benchmark~2:} Hybrid STAR-RIS with random active-element selection, while retaining optimized phases and transmit power levels from \textbf{Algorithm~\ref{alg:AO}}.
    \item \textbf{Benchmark~3:} Fully passive STAR-RIS with phase shift and power allocation design given by \textbf{Algorithm~\ref{alg:AO}}.
\end{itemize}
From Fig.~\ref{fig:EEM}, as the array size increases, all schemes show consistent EE improvement due to enhanced array gain. The proposed hybrid STAR-RIS design with activation ratios $\varrho=0.6$ and $\varrho=0.2$ achieves the highest EE across the antenna range, outperforming all benchmarks. Notably, $\varrho=0.2$ delivers the best overall EE, indicating that activating a moderate portion of elements offers an optimal trade-off between communication gains and additional power consumption, consistent with Fig. 4. \textbf{Benchmark~1} shows marginal EE growth with larger arrays, while \textbf{Benchmark~3} remains insensitive to BS antenna size due to the absence of adaptive optimization. \textbf{Benchmark~2} further reveals that random active-element selection significantly limits EE performance. These results confirm the \textbf{proposed} design’s superiority and underscore the need for efficient activation-ratio control in large-scale STAR-RIS systems.

Figure~\ref{fig:EESSINR} illustrates the EE versus the minimum sensing SINR threshold $\gamma_s^{\text{min}}$. The \textbf{TB} scheme applies tensor-based precoding via \textbf{Algorithm~\ref{alg:parafac}}, while \textbf{NP} serves as a non-precoding benchmark. \textbf{TB} consistently outperforms \textbf{NP} across all feasible $\gamma_s^{\text{min}}$ values, confirming the advantage of tensor-based precoding in jointly managing sensing and communication resources. Furthermore, with a fixed $\varrho=0.2$, the $N=49$ configuration achieves higher EE than $N=64$, as activating more elements raises power consumption, offsetting the rate gain. For the \textbf{TB} scheme, the EE exhibits a two-stage behavior with respect to $\gamma_s^{\text{min}}$. In the low-$\gamma_s^{\text{min}}$ regime, EE remains nearly flat, as the sensing SINR constraint is inactive or weakly active and can be satisfied without extra transmit power or STAR-RIS phase shifts. As $\gamma_s^{\text{min}}$ increases, the sensing constraint becomes active and dominates resource allocation, requiring additional power and phase adjustments to enhance the sensing link, leading to rapid EE degradation. For \textbf{NP}, adding RIS elements partly offsets the lack of precoding, giving better performance at $N=64$ than $N=49$. However, under strict sensing demands (high $\gamma_s^{\text{min}}$), \textbf{NP} becomes infeasible, while \textbf{TB} remains feasible over a wider range.

Figure~\ref{fig:fig8} shows the 3D RMSE for joint AoA and range estimation as a function of the sensing SINR, considering different array sizes and snapshot counts. All curves decrease monotonically with SINRs, reflecting the transition from a noise-dominated regime at low SINR to one where accuracy approaches the root-CRB at high SINRs.
Increasing the array size from $3\times3$ to $4\times4$ yields significant gains, especially in low and moderate SINR regions, due to the larger spatial aperture and narrower beamwidth, which enhance the MUSIC estimator’s angular resolution. Similarly, using more snapshots improves performance by making the sample covariance matrix a closer approximation of the true covariance, an effect most evident at medium SINRs.
At sufficiently high SINRs, all curves converge as angular estimation errors become negligible, leaving FMCW range resolution as the dominant factor. These results underscore the complementary roles of spatial aperture and temporal averaging in improving sensing accuracy for the considered ISAC system.

\begin{figure}[t]
    \centering
\includegraphics[width=0.41\textwidth]{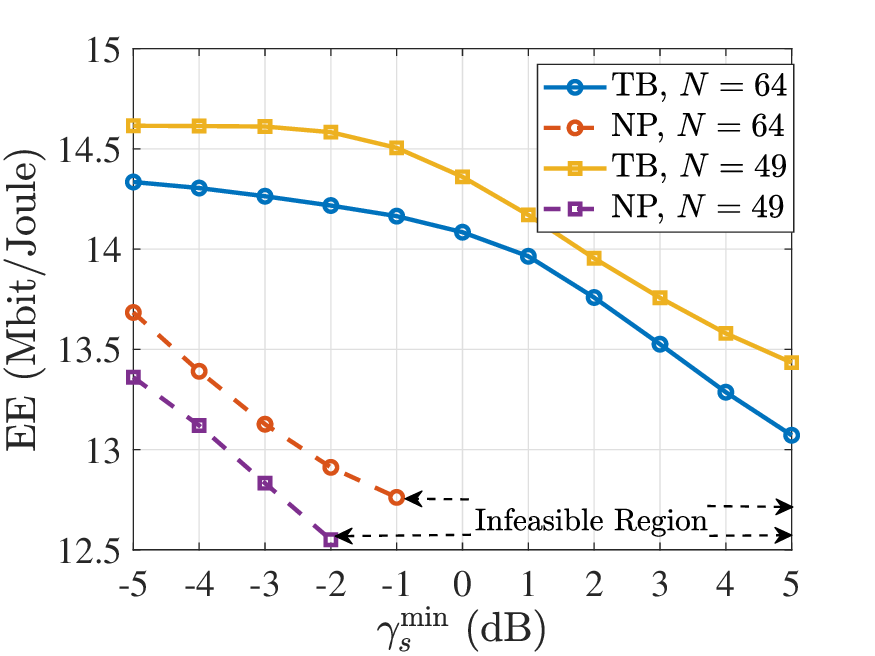}
\vspace{-0.4em}
    \caption{\small  EE versus the sensing SINR threshold ($M=36$,  $K_c=2$, $K_p=2$).}
    \label{fig:EESSINR}
    \vspace{-1em}
\end{figure}

\vspace{-0.5em}
\section{Conclusions} \label{sec:Conclusion}
We have developed a unified hybrid STAR-RIS-assisted framework to simultaneously support communication, localization, and WPT. A PARAFAC-based channel estimation method was introduced to decouple the cascaded channels with low training overhead, which enabled accurate tensor-aided precoding for both communication and sensing. To maximize the system EE, a joint optimization problem was formulated that incorporated communication QoS requirements, sensing SINR constraints for target detection, and nonlinear energy harvesting characteristics. This problem was addressed through an alternating procedure that combined fractional programming, SCA, and a multi-seed greedy strategy for active-element selection.
Numerical results verified that the proposed approach significantly improves the EE, while maintaining reliable sensing and power-transfer performance. The presented framework provided a scalable and effective solution for hybrid STAR-RIS-enabled integrated systems and offered valuable insights for future designs of multi-functional reconfigurable intelligent surfaces in beyond-5G and 6G networks.

\begin{figure}[t]
    \centering
\includegraphics[width=0.41\textwidth]{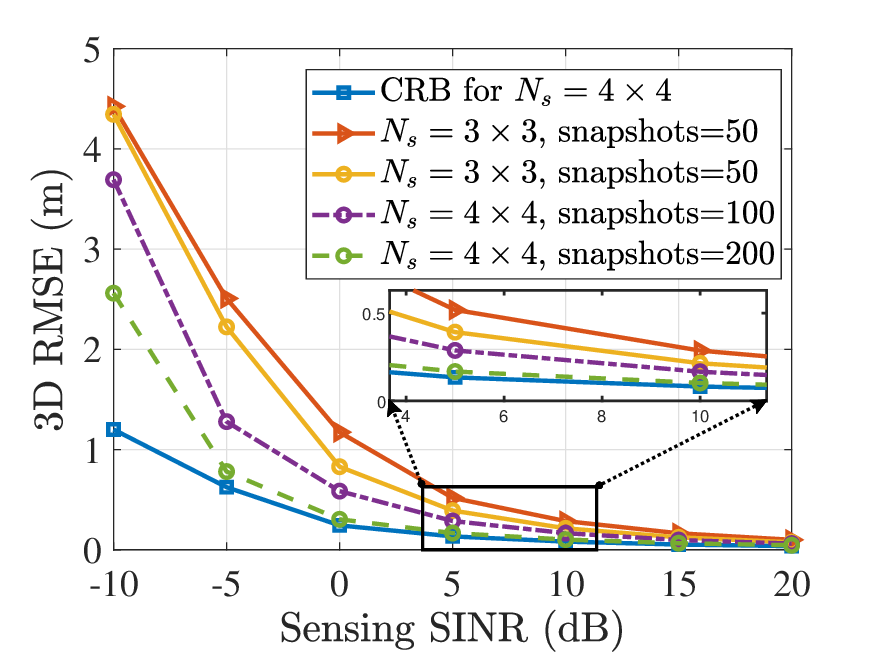}
\vspace{-0.5em}
        \caption{\small 3D RMSE versus the sensing SINR for varying array sizes and snapshot numbers ($M=48$, $N=64$, $K_c=2$, $K_p=2$).}
    \label{fig:fig8}
    \vspace{0.3em}
\end{figure}

\vspace{-0.2em}
\appendices
\vspace{-0.5em}
\section{Derivation of the FIM Elements\label{pro_CRB}}
By vectorizing the
received echo signal in~\eqref{ys_T}, we obtain
\begin{equation}
 \text{vec}(\qY_s)=\text{vec}(\qH_s\qS) + \text{vec}(\tilde{\qN}_s)  \sim \mathcal{CN}(\qu, \qR_n),
\end{equation}
where $\mathbf{u} =\text{vec}(\qH_s \qS)$ and $\qR_n=\sigma_s^2\qI_{N_sT}$\cite{He:RSS-Localization:IOT:2022},\cite{He:CF-RSS-Localization:TWC:2025}.
The $(i,j)$-th entry of the FIM is given by
\begin{equation}
    [\qJ]_{ij} = \frac{2}{\sigma_s^2}\Re\!\left\{
    \frac{\partial \mathbf{u}^H}{\partial {\zeta}_i}
    \frac{\partial \mathbf{u}}{\partial {\zeta}_j}
    \right\}, \quad\forall i,j \in \{1,2,3,4,5\}.
\end{equation}
 We further define $\boldsymbol{\vartheta}=[\theta,\psi]^T$ and  $\qU\!=\! \qb_s(\theta,\psi)\ \qa^T_\text{RT}(\theta_{\text{D}}^{\text{RT}}, \psi_{\text{D}}^{\text{RT}})\mathbf{\Theta}_r\qG$. Then, it follows that
\begin{equation}
    \frac{\partial\qu }{\partial\bm\zeta}\!=\!\!\big[
    \beta_s\text{vec}(\dot{\qU}_{\vartheta}\qS),     
        \beta_s\text{vec}(\qU\dot{\qS}),
        [1,j]\text{vec}(\qU\qS),
    \big]
\end{equation}
where $\dot{\qU}_{\vartheta}=\frac{\partial\qU}{\partial\vartheta}$ and $\dot{\qS}=\frac{\partial\qS}{\partial\tau}$. We define the covariance matrix of the transmit signal~\eqref{eq:signal}, as $\qR=\Ex\{\qx(t)\qx(t)^H\}$. The FIM matrix elements can be obtained as
\vspace{0.2em}
\begin{align}
{J}_{\vartheta\vartheta}
&=\frac{2|\beta_s|^2T}{\sigma_s^2}
\Re\big(\tr\big(\dot{\qU}_{\vartheta}\qR\dot{\qU}_{\vartheta}^H\big)\big), \vartheta\in\{\theta,\psi\},\\
J_{\vartheta\tau}&=\frac{2|\beta_s|^2T}{\sigma_s^2}\Re\big(\qS^H\dot{\qU}_{\vartheta}^H\qU\dot{\qS} \big),\vartheta\in\{\theta,\psi\},
\\
 J_{\vartheta\beta_R}&=\frac{2T}{\sigma_s^2}\Re\big(\beta_s\tr(\dot{A}_{\theta}\qR\qU^H)^H \big),\vartheta\in\{\theta,\psi\},
\\
J_{\vartheta\beta_I}&=\frac{2T}{\sigma_s^2}\Re\big(j\beta_s\tr(\dot{A}_{\theta}\qR\qU^H)^H \big),\vartheta\in\{\theta,\psi\},
\\
J_{\tau\tau}&=\frac{2|\beta_s|^2T}{\sigma_s^2}\Re\big(\tr(\qU^H\qU\dot{\qS}\dot{\qS}^H)\big),\\
J_{\beta_R\beta_R}&=J_{\beta_I\beta_I}=\frac{2T}{\sigma_s^2}\Re\big(\tr(\qU\qR\qU^H)\big),\\
J_{\beta_R\beta_I}&=0.
\end{align}

\vspace{-1em}
\bibliographystyle{IEEEtran}
\bibliography{Reference}

\vspace{-3em}
\begin{IEEEbiography}[{\includegraphics[width=1in,height=1.25in,clip,keepaspectratio]
{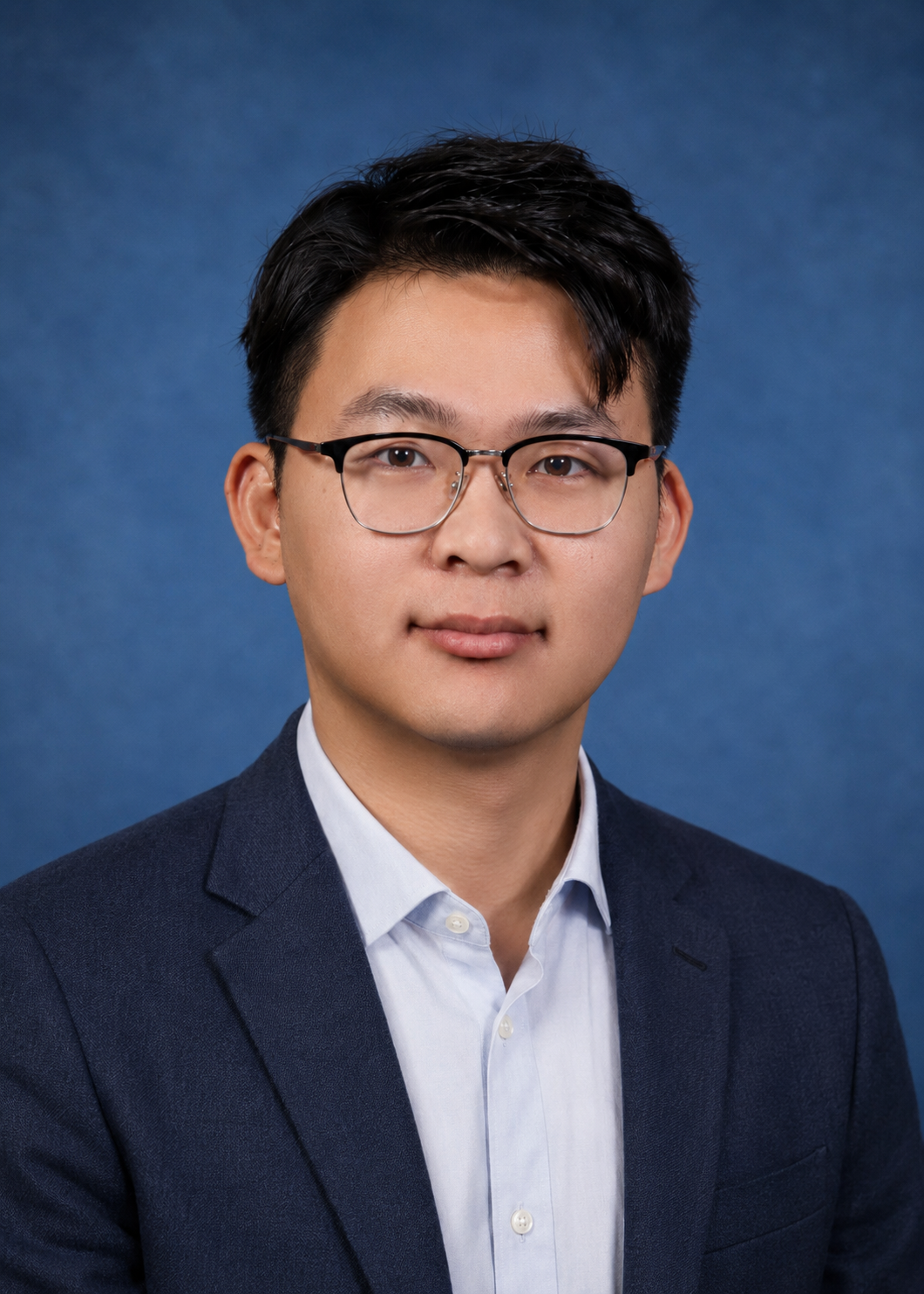}}]
{Haoran Ni} received the B.S. degree from the University of Huddersfield, U.K., in 2022, and the M.S. degree from Durham University, U.K., in 2023. He is currently pursuing the Ph.D. degree with the Centre for Wireless Innovation (CWI), Queen's University Belfast, U.K. His research interests include integrated sensing and communication (ISAC) and simultaneously transmitting and reflecting reconfigurable intelligent surfaces (STAR-RISs).
\end{IEEEbiography}

\vspace{-2em}
\begin{IEEEbiography}[{\includegraphics[width=1in,height=1.25in,clip,keepaspectratio]
{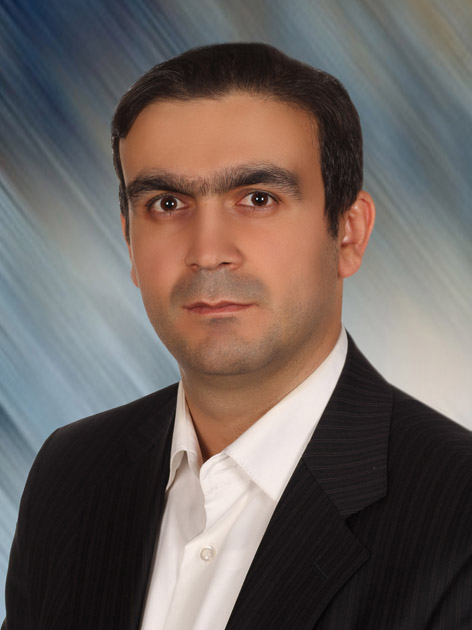}}]
{Mohammadali Mohammadi} (Senior Member, IEEE) is currently a Lecturer with the Centre for Wireless Innovation (CWI), Queen’s University Belfast, U.K. Previously, he held the position of
a Research Fellow with CWI from 2021 to 2024.
He has published more than 90 research papers in
accredited international peer-reviewed journals and
conferences in the area of wireless communication
and has co-authored three invited book chapters.
His research interests include signal processing for
wireless communications, cell-free massive MIMO,
integrated sensing and communications (ISAC), reconfigurable intelligent surfaces, and physical layer security. He received the Exemplary Reviewer Award from \textsc{IEEE Transactions on Communications} (2020 and 2022) and \textsc{IEEE Communications Letters} (2023) and the Exemplary Editor Award from \textsc{IEEE Communications Letters} in 2025. He serves as the Editor for the \textsc{IEEE Transactions on Communications}, \textsc{IEEE Communications Letters}, \textsc{IEEE Open Journal of the Communications Society}, the Digital Signal Processing, and the Physical Communication
(Elsevier).
\end{IEEEbiography}

\vspace{-3em}
\begin{IEEEbiography}[{\includegraphics[width=1in,height=1.25in,clip,keepaspectratio]
{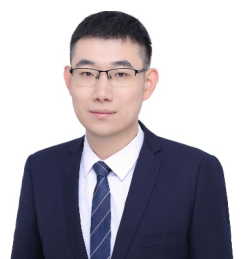}}]
{Xidong Mu} (Member, IEEE, \url{https://xidongmu.github.io/}) received the Ph.D. degree in Information and Communication Engineering from the Beijing University of Posts and Telecommunications, Beijing, China, in 2022. He was with the School of Electronic Engineering and Computer Science, Queen Mary University of London, from 2022 to 2024, where he was a Postdoctoral Researcher. He has been a lecturer (an assistant professor) with the Centre for Wireless Innovation (CWI), School of Electronics, Electrical Engineering and Computer Science, Queen’s University Belfast, U.K. since August 2024. His research interests include flexible-antenna technologies, reconfigurable surface aided communications, next generation multiple access (NGMA), integrated sensing and communications, and optimization theory. 

Xidong Mu is a Web of Science Highly Cited Researcher. He received the IEEE ComSoc Outstanding Young Researcher Award for EMEA region in 2023 and the IEEE ComSoc Wireless Communications Technical Committee (WTC) Outstanding Young Researcher Award in 2025. He is the co-recipient of the 2024 IEEE Communications Society Heinrich Hertz Award, the Best Paper Award in ISWCS 2022, the 2022 IEEE Signal Processing and Computing for Communications Technical Committee (SPCC-TC) Best Paper Award, and the Best Student Paper Award in IEEE VTC2022-Fall. He serves as the secretary of the IEEE ComSoc Technical Committee on Cognitive Networks (TCCN), the secretary of the IEEE ComSoc SPCC-TC, and the URSI UK Early Career Representative (ECR) for Commission C. He also serves as the Assistant to Editor-in-Chief of \textsc{IEEE Transactions on Green Communications and Networking} and an Editor of \textsc{IEEE Transactions on Communications}. 

\end{IEEEbiography}

\begin{IEEEbiography}[{\includegraphics[width=1in,height=1.25in,clip,keepaspectratio]{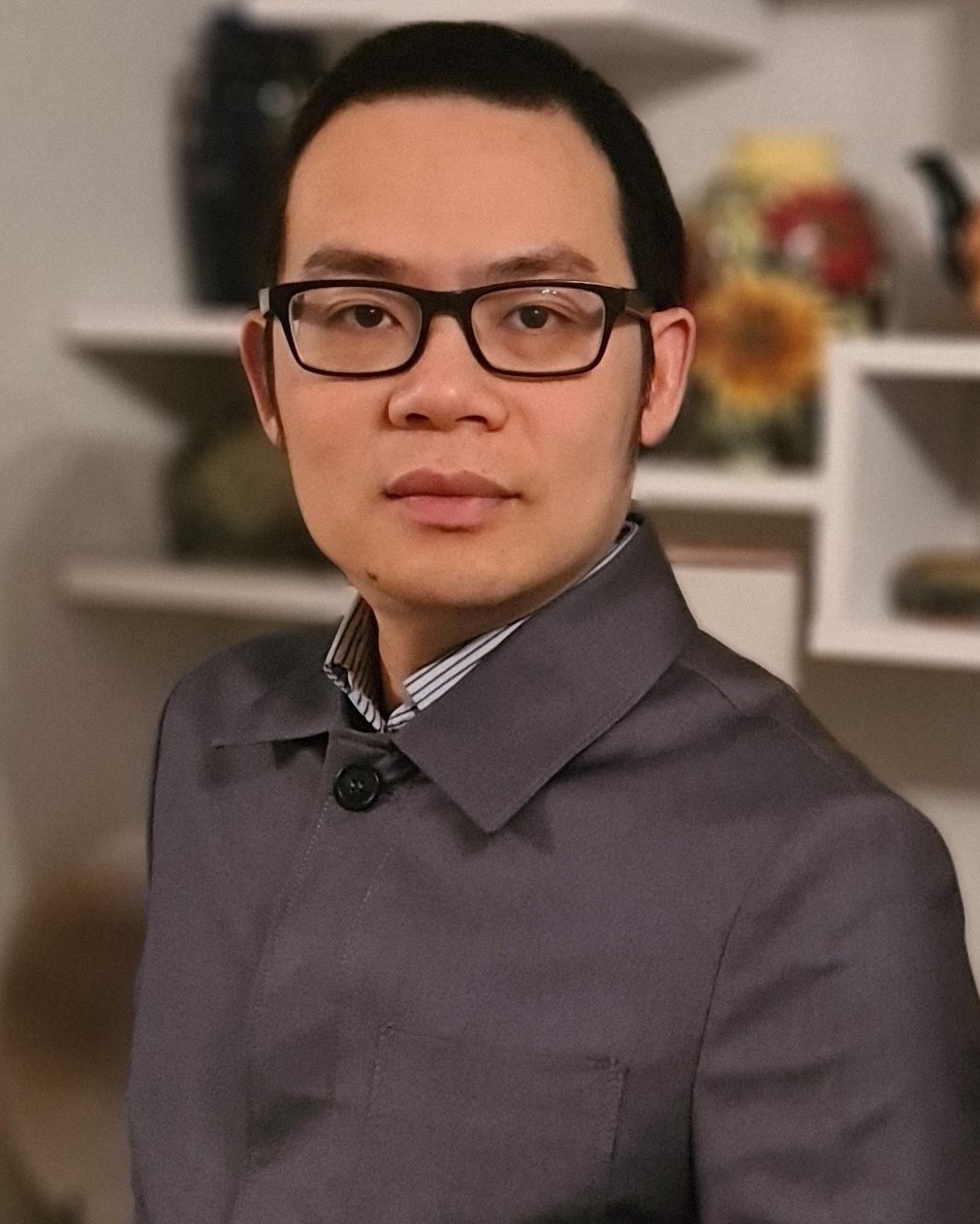}}]
{Hien Quoc Ngo} (Fellow, IEEE)  is currently a Reader with Queen's University Belfast, U.K. His main research interests include massive MIMO systems, cell-free massive MIMO, reconfigurable intelligent surfaces, physical layer security, and cooperative communications. He has co-authored many research papers in wireless communications and co-authored the Cambridge University Press textbook \emph{Fundamentals of Massive MIMO} (2016).

He received  the IEEE ComSoc Test of Time Paper Award for Advances in Communications in 2026,  the IEEE ComSoc Stephen O. Rice Prize in 2015, the IEEE ComSoc Leonard G. Abraham Prize in 2017, the Best Ph.D. Award from EURASIP in 2018, and the IEEE CTTC Early Achievement Award in 2023. He also received the IEEE Sweden VT-COM-IT Joint Chapter Best Student Journal Paper Award in 2015. He was awarded the UKRI Future Leaders Fellowship in 2019. He serves as the Editor for the IEEE Transactions on Wireless Communications, IEEE Transactions on Communications, the Digital Signal Processing, and the Physical Communication (Elsevier). He was an Editor of the IEEE Wireless Communications Letters, a Guest Editor of IET Communications, and a Guest Editor of IEEE ACCESS in 2017.
\end{IEEEbiography}

\vspace{-2em}
\begin{IEEEbiography}[{\includegraphics[width=1in,height=1.35in,clip,keepaspectratio]{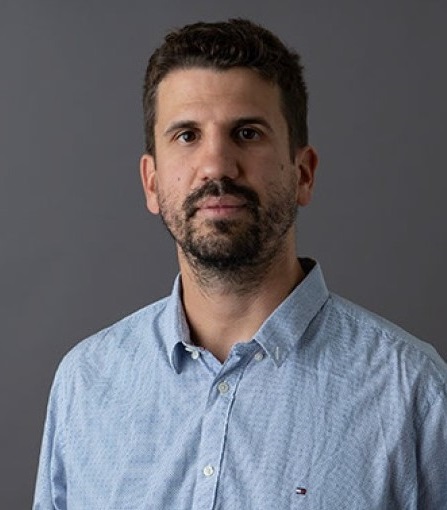}}]
{Michail Matthaiou}(Fellow, IEEE) obtained his Ph.D. degree from the University of Edinburgh, U.K. in 2008. 
He is currently a Professor of Communications Engineering and Signal Processing and Deputy Director of the Centre for Wireless Innovation (CWI) at Queen’s University Belfast, U.K. He is also an Eminent Scholar at the Kyung Hee University, Republic of Korea. He has held research/faculty positions at Munich University of Technology (TUM), Germany and Chalmers University of Technology, Sweden. His research interests span signal processing for wireless communications, beyond massive MIMO, reflecting intelligent surfaces, mm-wave/THz systems and AI-empowered communications.

Dr. Matthaiou and his coauthors received the IEEE Communications Society (ComSoc) Leonard G. Abraham Prize in 2017. He currently holds the ERC Consolidator Grant BEATRICE (2021-2026) focused on the interface between information and electromagnetic theories. To date, he has received the prestigious 2023 Argo Network Innovation Award, the 2019 EURASIP Early Career Award and the 2018/2019 Royal Academy of Engineering/The Leverhulme Trust Senior Research Fellowship. His team was also the Grand Winner of the 2019 Mobile World Congress Challenge. He was the recipient of the 2011 IEEE ComSoc Best Young Researcher Award for the Europe, Middle East and Africa Region and a co-recipient of the 2006 IEEE Communications Chapter Project Prize for the best M.Sc. dissertation in the area of communications. He has co-authored papers that received best paper awards at the 2018 IEEE WCSP and 2014 IEEE ICC. In 2014, he received the Research Fund for International Young Scientists from the National Natural Science Foundation of China. He is currently the Editor-in-Chief of Elsevier Physical Communication, a Senior Editor for \textsc{IEEE Wireless Communications Letters} and \textsc{IEEE Signal Processing Magazine}, an Area Editor for \textsc{IEEE Transactions on Communications} and Editor-in-Large for \textsc{IEEE Open Journal of the Communications Society}. 
\end{IEEEbiography}

\end{document}